\documentclass{amsart}

\usepackage{hyperref}
\usepackage{url}
\usepackage{amsmath,amssymb,amsthm}
\usepackage{fancyvrb}
\usepackage{enumerate} 
\usepackage{cases} 
\usepackage{color} 
\usepackage{graphicx}
\usepackage{bm}
\usepackage{algorithm}
\usepackage{algpseudocode}
\usepackage{geometry}
\usepackage{tikz}
\usetikzlibrary{automata,arrows,positioning,calc}

\numberwithin{equation}{section} \numberwithin{figure}{section}

{\theoremstyle{remark} \newtheorem{remark}{Remark}}
\newtheorem{definition}{Definition}[section]

\newtheorem{theorem}{Theorem}[section]

\newtheorem{lemma}{Lemma}[section]

\newcommand{\eps}{\varepsilon}
\newcommand{\Pt}{\tilde{P}}
\newcommand{\etat}{\tilde{\eta}}

\newcommand{\calN}{\mathcal{N}}
\newcommand{\calP}{\mathcal{P}}
\newcommand{\calA}{\mathcal{A}}
\newcommand{\calB}{\mathcal{B}}

\newcommand{\N}{\mathbb{N}}
\newcommand{\R}{\mathbb{R}}

\DeclareMathOperator*{\argmax}{arg\,max}

\title[Generalizing Liquid Democracy to multi-agent delegation]{Generalizing Liquid Democracy to multi-agent delegation: A Voting Weight Measure and Equilibrium Analysis}

\author[F.M.~Bersetche]{Francisco M.~Bersetche}
\address[F.M.~Bersetche]{Departamento de Matem\'atica, FCEyN, Universidad de Buenos Aires, Ciudad Universitaria 1428, Buenos Aires, Argentina.}
\address[F.M.~Bersetche]{Departamento de Matem\'atica, Universidad T\'ecnica Federico Santa Mar\'ia, Av. España
1680, Valpara\'iso, Chile.}

\email{bersetche@gmail.com}
\thanks{ The author thanks the anonymous reviewers for their careful reading of the manuscript and for their constructive observations, which have helped to clarify and improve the work. The author was partially supported by ANID through FONDECYT Project 3220254. Part of this work was performed in Centro At\'omico Bariloche under a postdoctoral fellowship funded by CONICET}

\begin{document}

\begin{abstract}

   In this study, we propose a generalization of the classic model of liquid democracy that allows fractional delegation of voting weight, while simultaneously allowing for the existence of equilibrium states. Our approach empowers agents to partition and delegate their votes to multiple representatives, all while retaining a fraction of the voting weight for themselves. We introduce a penalty mechanism for the length of delegation chains. We discuss the desirable properties of a reasonable generalization of the classic model, and prove that smaller penalty factors bring the model closer to satisfying these properties. In the subsequent section, we explore the presence of equilibrium states in a general delegation game utilizing the proposed voting measure. In contrast to the classical model, we demonstrate that this game exhibits Nash equilibria, contingent upon the imposition of a penalty on the length of delegation chains.

\end{abstract}

\maketitle

\section{Introduction}\label{sec:intro}

Liquid democracy (LD) is a collective decision making method that lies between representative democracy and direct democracy. In LD, an agent can choose to either vote directly or to select another agent as their proxy. Unlike classic proxy voting, the delegation is transitive; this means that the proxy may delegate their voting weight further, creating delegation paths along which the voting weight flows and accumulates.

 Over the past decade, several practical implementations of LD have been conducted \cite{10years}. Notable examples include Zupa \cite{zupa}, which was employed by the Student Union of the Faculty of Information Studies in Novo mesto, Slovenia; LiquidFeedback \cite{LF}, which was used by the German Pirate Party; and Sovereign \cite{siri}[Section 2], a blockchain-based platform developed and utilized by the Democracy Earth Foundation in Argentina. Of these implementations, the use of LiquidFeedback by the German Pirate Party is the best documented in the literature and will be considered as a benchmark for LD practical applications throughout this work. 
  
Proxy voting models have been studied since Carroll's work in 1884 \cite{1884}. In the last decade, these models have been investigated by the political science \cite{green} and artificial intelligence communities \cite{brandt,chris_grossi}. Subsequent research has focused on potential improvements and modifications, as well as possible weaknesses of these models \cite{potentials_prob,brill_talmon,critique,indefense,procaccia}, among others. For a comprehensive review of this topic see \cite{10years} and the references therein.

Recently, there has been an increasing interest in examining the existence of equilibrium states in these systems from a game theory perspective \cite{grossi,game}. Particularly, in \cite{game} the authors establish that the classical model of liquid democracy generally lacks equilibrium states. To address this problem, one may consider a model that allows delegation to multiple agents through the use of mixed strategies. However, while this model ensures the existence of equilibrium states, it becomes evident that it does not preserve desirable properties for a liquid democracy model. 

The aim of this work is to provide a model of liquid democracy that ensures the existence of equilibrium states while simultaneously preserving desirable properties in a liquid democracy system. The proposed model also generalizes the single delegation paradigm by allowing delegation to multiple agents.

\subsection*{Contribution}

The prevalent approach in liquid democracy implementations and existing literature typically adopts a simple delegation paradigm, where each agent has the option to delegate their entire voting weight to a single agent or cast a direct vote \footnote{ Sovereign \cite{siri} introduces a multi-agent delegation system, granting users a finite number of independently delegatable votes; however, each received vote must be independently dispensed by each proxy.}. This paper generalizes this model by allowing agents to divide their voting weight arbitrarily and delegate it to multiple agents, while also retaining part of it for themselves, and enabling the existence of equilibrium states. This generalization is achieved through a novel voting weight measure that preserves the main properties of classic delegation. The introduced measure incorporates a penalty factor based on the length of delegation chains, allowing the existence of equilibrium states. As in most platforms and available works, we consider the voting weight trapped in delegation cycles as null votes (see \cite{mythordisaster}).

Throughout this work, we posit that a reasonable generalization of the classical LD model should fulfill, at least, the following properties:

\begin{itemize}

    \item[(P1)] It generalizes the concept of simple delegation, whereby classic delegation is recovered when agents choose to delegate their voting weight without dividing it. 

    \item[(P2)] Delegation makes sense. If agent $a$ delegates her voting weight to agent $b$, the distribution of voting weight is equivalent to what would be obtained if agent $a$ replicated the delegation weighting of agent $b$. This implies that delegating to an agent is equivalent to replicating their actions. Although we are being advisedly vague here, this property is formalized in Definition \ref{def:P2}. 
    
    \item[(P3)] If an agent chooses to retain a fraction $q \in (0,1]$ of the vote for herself, then her voting weight should be no less than $q$. Conversely, if an agent opts to delegate all the weight, her voting weight should be rendered null. This property is referred to as \textit{self-selection}.

\end{itemize}
Note that property (P2) implies that, for an agent, delegating is equivalent to allowing their proxy to make decisions on their behalf. That is, starting from a delegation setting where agent $a$ delegates their vote to agent $b$, if a new setting is constructed where agent $a$ copies the delegation weightings of agent $b$, the resulting voting weight should remain unchanged. This intuitively corresponds with the expected outcome of delegation. We refer to (P2) as the \textit{delegation property}. We emphasize that the properties (P1)-(P3) are desirable, but do not necessarily provide a comprehensive characterization of a reasonable generalization of the classic LD model. In this work we establish the fulfillment of properties (P1)-(P3) by the proposed voting measure within a defined margin of error determined by the penalty factor. 

The evolution of the decision making process is a crucial factor to consider in our analysis. As reported in \cite{temporal}, the experience with LiquidFeedback reveals that the behavior of agents changes over time due to their observations of other agents' actions, which is essential for the development of decision-making processes. Therefore, we assume that an agent's delegation weighting evolves over time and depends on real-time feedback received by the agent. This leads us to ask whether this process can reach an equilibrium state, then in the second part of this work, we introduce a game-theoretic model where agents serve as players and their delegation weightings represent their strategies. Each agent $i$ is endowed with a utility function that measures the amount of their own vote attributed to a specific player $j$, weighted by a satisfaction factor which accounts for $i$'s contentment with the proportion of their voting weight ceded to $j$. We demonstrate the existence of Nash equilibria in this game under the condition of non-zero penalties imposed on extended delegation chains. Furthermore, we present the main result of this work, where we establish the possibility of defining a voting weight measure such that the fundamental properties of classic delegation are upheld within an acceptable margin of error, while simultaneously attaining Nash equilibria for the examined game. 
To the best of our knowledge, this study offers the first formal proof establishing the feasibility of a liquid democracy system that guarantees the existence of equilibrium states.

\subsection*{Related work} 
Regarding models that allow multi-agent delegation, we can mention the model proposed in \cite{degrave}, which is similar to the model presented here when the penalty factor is not considered. Spectral type centrality measures have been proposed for measuring voting weight; \cite{yoshida} explores the use of PageRank as a voting measure, while \cite{viscousDemocracy} propose another voting measure based on a PageRank-type centrality which seeks to penalize long delegation paths. Although simple delegation is used here for theoretical analysis purposes, these techniques can easily be generalized to multi-agent delegation.
On the other hand, following a different approach, in \cite{ugrandi}, the authors propose and analyze a model based on multi-agent ranked delegation, in which agents provide a ranking over multiple delegations to be used as backups when solving possible delegation cycles. Similarly, \cite{brill} and \cite{brill_talmon} study the delegation of vote through a ranking of preferences, while \cite{fluid} proposes a technique based on fluid dynamics that attempts to minimize concentration of votes by choosing an appropriate delegation graph according to the agents' preferences.

In Section \ref{sec:nash}, we define a game in which the delegation weighting of each agent is considered as a strategy, and a utility function based on the preferences of each agent over the others. This type of game has been first studied in \cite{grossi}, with regard to simple delegation. Particularly relevant to our work is \cite{game}, where the authors analyze the existence of equilibrium states in this type of game under various optimality criteria. The problem discussed in Section \ref{sec:nash} can be viewed as a generalization of the \textbf{EXISTENCE} problem \footnote{ This problem focuses on establishing the existence of a Nash equilibrium in a game where players' strategies are determined by their voting behavior, and their utility functions are governed by their individual preferences (see Section \ref{sec:EXISTENCE}). } defined in \cite[Section 3.2]{game}. In \cite{zhang_grossi2}, the authors utilize a weighted delegation scheme based on mixed strategies to examine the truth-tracking problem. This scheme is analogous to the one explored in Section \ref{subsec:apMS}. Other works addressing similar topics include \cite{pirates,zhang_grossi}.

\subsection*{Organization of the paper} The rest of the paper is organized as follows. Section \ref{sec:preli} provides an intuitive description of the proposed voting weight measure, beginning with its probabilistic interpretation. These ideas are formalized in Section \ref{sec:formal}, where we define the measure and prove its main properties. In Section \ref{sec:nash}, game theory is employed to investigate the existence of equilibrium states in the context of agents' ability to modify their delegation weighting according to system outcomes. Ultimately, our main result is demonstrated, which establishes the feasibility of constructing a voting weight measure through imposition of penalties on long delegation chains, that both preserves classic delegation properties and allows for Nash equilibria. 

\section{Preliminaries and main ideas}\label{sec:preli}

\subsection{Preliminaries}\label{sec:subpreli}
 
Let $N = \{1,...,n\}$ be the set of agents. Each agent $i$ expresses to whom they want to delegate their voting weight, and the fraction they want to delegate, through the delegation weighting $x_i \in X$, with $X := \{ x \in \mathbb{R}^{n}_{\geq 0} : \sum_k x_k = 1\}$. Here $x_{ij}$ represents the fraction of voting weight that the agent $i$ wants to delegate to agent $j$. We define the matrix $P:= (x_1|...|x_n) \in \mathbb{R}^{n \times n}$ with the delegation weightings as columns, and we call it delegation matrix. Clearly $P$ is a (left)stochastic matrix (since $\sum_{j}P_{ji} = 1$), with the delegation weighting of agent $i$ stored in its $i$-th column. We denote by $\mathcal{P}^n\subset \mathbb R^{n\times n }$ the set of all possible delegation matrices. Note that while the delegation weightings from set X and delegation matrices from $\mathcal P^n$, are strongly related, it is convenient to treat these two concepts separately, especially in Section \ref{sec:nash}, where delegation weightings represent agents' strategies.

Given $i$ and $j \in N$, we say that there is a delegation path between $i$ and $j$ if there is a strictly positive sequence $(x_{ik_1},x_{k_1k_2},...,x_{k_{q-1} k_{q}},x_{k_qj})$ with $k_1,...,k_q \in N$. Given a vector $v \in \R^n$, we define the matrix $diag(v) \in \R^{n \times n}$ as $diag_{ii}(v) := v_i$ for all $1 \leq i \leq n$, and $diag_{ij}(v) := 0$ for all $i \not = j$. As is common in packages designed for numerical methods, if $A \in \R^{n \times n}$, we define $diag(A) \in \R^{n}$ as $diag_i(A) := A_{ii}$ for all $1 \leq i \leq n$. 
 Unless specified otherwise, throughout this paper vectors will always be regarded as column vectors.

\subsection{Motivation}\label{sub:problem_mutidel}

Let $N = \{1,...,n\}$ denote the set of agents in the classic model of liquid democracy (LD), where agents can delegate their voting weight to a single agent, including themselves, and we refer to agents who preserve all their voting weight as \textit{candidates}. The delegation matrices obtained under this paradigm correspond to adjacency matrices associated with delegation graphs defined by agent preferences. Defining $\calA^n := \{ A \in \calP^n : A_{ij} \in \{0,1\} \text{ for all } 1 \leq i,j \leq n \}$, the classic measure of voting weight is a function $V^c: \calA^n \to \R^n_{ \geq 0}$, where $V^c_i(A)$ represents the voting weight of agent $i$ under the delegation configuration given by the matrix $A \in \calA^n$. In the classic model of LD, the number of votes received by a candidate is determined as the sum of votes received by all agents who delegate their weight to the candidate, plus the votes of all agents who delegate to these agents, and so on. This relationship can be expressed recursively as follows: let $D_i$ denote the set of agents who choose to delegate their vote to agent $i$. The voting weight of agent $i$ is then determined by the relationship $V_i = \sum_{j \in D_i} V_j$.
 Although this procedure can be easily described recursively, we provide an equivalent vectorial definition, which is more manageable in the context of this work. To this end, for a given $A \in \calA^n$, we can provide a vectorial definition of $V^c: \calA^n \to \R^n_{ \geq 0}$ as follows:
\begin{equation}\label{eq:classicLD}
    V^c(A) := \lim_{k\to \infty} ( A^k \mathbf{1} ) \odot diag(A),
\end{equation}
where $\mathbf{1} = (1,...,1) \in \R^n$ and $\odot$ denotes the element wise product. Here $c$ stands for \textit{classic} \footnote{Although in the case $A \in \calA^n$ the limit in \eqref{eq:classicLD} converges after $n$ steps, this formula is readily generalizable to the case $A \in \calP^n$.}.

To illustrate the concept, consider the example shown in Figure \ref{fig:example_0}. In this example, the classical model yields a voting weight of 3 for agent 2, while all other agents have a voting weight of 0. Letting $A$ denote the adjacency matrix of this example, it is straightforward to verify that the entries of the vector $A^k \mathbf{1}$ corresponding to agents 1, 2, and 3 converge to the expected values starting from $k = 2$. In contrast, the entries corresponding to the other agents, who are part of a delegation cycle, oscillate as $k \to \infty$. These entries, however, are multiplied by zero as indicated in \eqref{eq:classicLD}, resulting in the expected outcome for the classical LD model.

\begin{figure}[ht]
\begin{center}
	\begin{tikzpicture}[->, >=stealth', auto, semithick, node distance=1cm]
	\tikzstyle{every state}=[fill=white,draw=black,thick,text=black,scale=0.7]
	\node[state] (B) at (-1,0.5) {$1$};
	\node[state] (C) at (0.5,2) {$2$};
	\node[state] (C2) at (2,0.5) {$3$};
	\node[state] (E) at (3.5,0.5) {$4$};
	\node[state] (F) at (5,2) {$5$};
	\node[state] (F2) at (6.5,0.5) {$6$};
	\path

	(B) edge[]	node{}	(C)
	(C2) edge[]	node{}	(B)
        (C) edge[loop above]	node{}	(C)
	(E) edge[bend left,above]	node{}	(F)
        (F2) edge[above]	node{}	(F)
        (F) edge[bend left,below]	node{}	(E);
	\end{tikzpicture}
\end{center}
\caption{An example illustrating a simple delegation configuration in the classical LD model.}
\label{fig:example_0}
\end{figure}

In general, given the properties of adjacency matrices, it can be readily observed that successive products $A^k \mathbf{1}$ direct and accumulate votes from agents to candidates through delegation chains. Note that $diag_i(A) \in \{0,1\}$ for all $i$, and $diag_i(A) = 1$ if and only if $i$ is a candidate. Thus, the element-wise product with $diag(A)$ preserves the entries of $A^k\mathbf{1}$ corresponding to candidates, while nullifying the remaining entries. This ensures that the limit in \eqref{eq:classicLD} is well-defined, as votes trapped in delegation cycles are multiplied by zero. Writing the recursive relationship of the classical model in matrix form, the equivalence with the formula proposed in \eqref{eq:classicLD} can be easily confirmed.

We now consider multi-agent delegation, that is, we aim to extend the function $V^c$ to the entire set $\calP^n$ in such a way that this function constitutes a reasonable generalization of the classic LD model. As a first idea one can use the definition \eqref{eq:classicLD} by considering $A \in \calP^n$, since the limit in this expression remains well-defined. This generalization is, in fact, closely related to the Eigenvector Centrality Measure \cite{BONACICH}, and related in general to spectral centrality measures, as the PageRank used as voting measure in \cite{castillo,viscousDemocracy,yoshida}. However, it is possible to see that this generalization does not satisfy the properties stated in the previous section. However, it should be mentioned that if we prohibit agents from preserving a fraction of votes for themselves that is different from 1 or 0, the aforementioned generalization preserves the properties stated in the previous section. 
      
In \cite{game}, it is shown that the classic LD model typically lacks equilibrium states in the context of delegation games. However, Nash equilibria can be achieved by incorporating mixed strategies, suggesting a potential multi-agent delegation model. Unfortunately, this model does not possess the delegation property (P2) (see \ref{subsec:apMS}), making it unsuitable as a generalization of the classic model. Our goal is to establish a measure of multi-agent voting weight to ensure equilibrium states while maintaining properties (P1) to (P3). To achieve this, we propose a generalization of the classic model by introducing penalties on the delegation chain length. Specifically, we demonstrate that as the penalty approaches zero, a voting weight measure that satisfies properties (P1) to (P3) is attained. However, like the classic model, the acceptance of equilibrium states by this limiting case voting measure is not guaranteed in general (see \ref{sec:app_V_not_stable}). Nonetheless, we show that with a positive penalty factor, the model achieves at least one Nash equilibrium in the delegation game. Moreover, properties (P1) to (P3) are upheld within a defined margin of error depending on the penalty factor. These findings are formally presented and summarized in Theorem \ref{teo:main}, which is our main result. 

\subsection{The voting weight measure} \label{sec:particles}
 Now our goal is to construct a function $V: \calP^n \to \R^{n}$ that represents the voting weight of each agent $i$, which is a reasonable generalization of $V^c$. In Section 3, we provide a formal definition of the proposed function $V$, as well as some alternative definitions, but in this section, we first introduce its probabilistic interpretation. Specifically, we construct $V$ using a particle system that depends on the agents and their delegation weightings, governed by the following rules:
\\

Given a delegation matrix $P \in \calP^n$, within a time interval of length $\delta t>0$: 
\begin{itemize}
    \item Each agent $j$ receives 1 particle with probability $\delta t$.
    \item A particle in position $j$ has a probability $\delta t P_{ij}$ to jump to position $i$, with $j\not=i$.
    \item A particle in position $j$ leaves the system with probability $\delta t P_{jj}$. In such a case we shall say that the particle was consumed by agent $j$. 
\end{itemize}
 We consider the limit as $\delta t \to 0$, and define $V_i(P)$ as the mean rate of particles consumed by agent $i$ when the system reaches equilibrium.

 To compute $V_i(P)$, we introduce $u_i$ as the average number of particles occupying position $i$ in the equilibrium state. We can then estimate $V_i(P)$ as the expected number of particles consumed by agent $i$ per unit time, given by $P_{ii} u_i$. To determine $u_i$, we set up a set of balance equations for the particle system described in the previous section. Using the rules governing the system, the evolution of the number of particles in position $i$ is given by: 
\begin{align*}
  u_i(t+\delta t) &= u_i(t) + \delta t \Big( \sum_{j \not = i} u_j(t) P_{ij} + 1 - u_i(t) \Big), \\
  \frac{u_i(t+\delta t) - u_i(t)}{\delta t} &= \sum_{j \not = i} u_j(t) P_{ij} + 1 - u_i(t). 
\end{align*}	
Since we look for $u_i$ at the equilibrium state, we have 
\begin{equation} \label{eq:prelim_equi}
    0 = \sum_{j \not = i} u_j P_{ij} + 1 - u_i.
\end{equation}
The above expression can be written in matrix form:
$$0 = \Pt u - u + \mathbf{1},$$
\begin{equation}\label{eq:prelim_u}
    u = (I - \Pt)^{-1}\mathbf{1},
\end{equation}
with $\Pt = P - diag( (P_{11},...,P_{nn}) )$. 

At this point, we can verify that $V(P)$ satisfies some of the properties enlisted in the preceding section. In particular, part of the property (P3) can be readily established, as when an agent $i$ delegates all their voting weight, we get $P_{ii}=0$, and thus $V_{i}(P) = P_{ii} u_i = 0$. Furthermore, we can observe that to determine which agents consume a part of the voting weight of a certain agent $i$, we may use the same particle system defined above, with the constraint that the influx of particles is only through the agent $i$. That is, we calculate $V(P)$ from $(I - \Pt)^{-1}\bm{\delta}_i$, where $\bm{\delta}_i \in \R^ {n}$, $\bm{\delta}_{ii} = 1$ and $\bm{\delta}_{ij} = 0$ if $i \not = j$. Consequently, $V_j(P)$ will quantify the amount of voting weight consumed by agent $j$, coming from agent $i$.

\subsection{Solving delegation cycles}

It is clear that equation \eqref{eq:prelim_u} is not always solvable. This is due to the possible existence of delegation cycles. We formally define a delegation cycle as follows:
\begin{definition}[Delegation cycle]\label{def:cycle}
Given a set of agents $N = \{1,...,n\}$ and their delegation weightings $x_1,...,x_n$, we say that $C \subseteq N$ is a delegation cycle if $x_{ii} = 0$ for all $i \in C$, and given $i \in C$, $x_{ij}=0$ for all $j \in N \setminus C$. 
\end{definition}
Note that a delegation cycle may also contain agents that merely delegate into it; the set $N_3$ introduced in \eqref{eq:N3} below is the largest delegation cycle. Based on this definition, it follows that for a cycle $C$, any particle that reaches an agent $i \in C$ cannot exit the system. To address this issue, we propose the introduction of an artificial agent $n+1$ into the system, where each particle has a probability $\delta t \eps$ of transitioning to agent $n+1$ from any position within a time interval $\delta t$, where $0<\eps<1$ and $P_{n+1 n+1} = 1$. Specifically, we add agent $n+1$ with delegation weighting $x_{n+1}=(0,...0,1) \in \R^{n+1}$, and redefine the weightings of the remaining agents as $x_i := ( (1-\eps)x_i, \eps )$, where $i \in N$. It is worth noting that every particle that reaches agent $n+1$ is absorbed. We denote the voting weight measure obtained through this method as $V^{\eps}$.

Following the modification introduced, delegation cycles no longer pose a challenge, as any particle trapped in a cycle is consumed by agent $n+1$. However, this modification leads to the penalization of extended delegation paths. To illustrate, consider a classic delegation chain of length $k$ between agent $i$ and agent $j$. For a particle originating from agent $i$, the probability of exiting the system through agent $n+1$ before reaching agent $j$ is given by $p^{ k}_{\eps} = 1 - (1 - \eps)^k > 0$. Consequently, the voting weight of agent $i$ dissipates through agent $n+1$, which results in the penalization of long delegation paths \footnote{This damping factor is similar to the one employed in PageRank computation, and is examined in \cite{castillo, viscousDemocracy}}. Moreover, $V_{n+1}^{\eps}$ denotes the aggregate of all the voting weight dissipated through delegation cycles and extended delegation paths. Note that $p^{k}_{\eps} \to 0$ as $\eps \to 0$, yet the voting weight dissipated by delegation cycles remains invariant. To illustrate this, classic delegation cycles can be considered infinite-length delegation chains, and $p^{\infty}_{\eps} := \lim_{k \to \infty} p^{k}_ {\eps} = 1$, which establishes the independence of $\eps$.
In order to eliminate all penalties in long but finite chains, we may define
\begin{equation}\label{eq:def_vp_nonf}
V(P) := \lim_{\eps \to 0} V^{\eps}(P),
\end{equation}
where $V_{n+1}(P)$ accumulates the total amount of voting weight lost in delegation cycles.

\begin{remark}

It is worth mentioning that $V$ can be constructed in a much simpler manner. For instance, one can use a straightforward generalization of the model proposed in \cite{degrave}, or alternatively, a construction based on Markov chains (see Section \ref{sec:alternative}). Starting from these simpler constructions, it can be easily shown that the measure $V$ satisfies properties (P1)-(P3). However, it can be proven that using this measure, in the context of the delegation game defined in Section \ref{sec:nash}, it is not true that an equilibrium state always exists, similar to what occurs in the classic delegation model (see Section \ref{sec:app_V_not_stable}).
As mentioned earlier, the main contribution of this work focuses on the definition and analysis of the measure $V^{\eps}$. We will see below that this measure possesses the following properties:

\begin{itemize}
    \item $V^{\eps} \to V$ when $\eps \to 0$, meaning the limiting case defines a reasonable voting measure,
    \item $V^\eps$ has at least one equilibrium state whenever $\eps >0$,
    \item $V^\eps$ satisfies properties (P1)-(P3) within a certain margin of error depending on $\eps$.
\end{itemize}

\end{remark}

\section{Formal framework} \label{sec:formal}

Let $N = \{1,...,n\}$ be a set of agents with delegation weightings $x_1,...,x_n$, and $P$ their delegation matrix. We define 
\begin{equation*}\label{eq:Peps}
    P^{\eps} := 	\begin{pmatrix} 
	(1-\eps)P & \mathbf{0} \\
	\bm{\varepsilon} & 1 \\
	\end{pmatrix}, 
\end{equation*}
\begin{equation}\label{eq:Ueps}
    u^{\eps} := (I - \Pt^{\eps})^{-1} \bm{1}_0,
\end{equation}
and
\begin{equation}\label{eq:VPeps}
    V^{\eps}(P) := (P^{\eps}_{11}u^{\eps}_1,...,P^{\eps}_{nn}u^{\eps}_n,u^{\eps}_{n+1}),
\end{equation}
where $\bm{\eps} = (\eps,...,\eps) \in \R^{1 \times n}$, $\tilde{P}^{\eps} = P^{\eps} - diag \big( (P_{11}^{\eps},...,P_{nn}^{\eps},P_{n+1 n+1}^{\eps}) \big)$ and $\bm{1}_0 = (1,...,1,0) \in \R^{n+1}$. Note that, since $\eps > 0$ we have $\rho(\Pt^{\eps}) = (1-\eps)\rho(\Pt) \leq 1-\eps < 1$, as $\Pt^\eps$ is block lower triangular with diagonal blocks $(1-\eps)\Pt$ and $0$, as a consequence $(I - \Pt^{\eps})^{-1}$ exists and $V^{\eps}(P)$ is well defined. 

We define the voting weight $V(P)$ associated to a delegation matrix $P$ as
\begin{equation}\label{eq:VP}
    V(P) := \lim_{\eps \to 0} V^{\eps}(P). 
\end{equation}
\subsection{Well-definedness}

To establish the existence of the limit in \eqref{eq:VP}, it is necessary to introduce the following auxiliary result.

\begin{lemma}\label{lemma:sequences}
Let $\{a_k\}_{k \in \N_0}$ and $\{b_k\}_{k \in \N_0} \in \ell^{\infty}$ be two uniformly bounded sequences, where $a_k$ is periodic with period $q \in \N$, and $\sum_k | b_k | < \infty$. Then:
\begin{itemize}
    \item[$(i)$] $\lim_{\eps \to 0} \sum_{k \geq 0} \eps(1-\eps)^k b_k = 0,$ 
    \item[$(ii)$] $\lim_{\eps \to 0} \sum_{k \geq 0} \eps(1-\eps)^k a_k = \ell$, for some $\ell$ depending on $\{a_k\}_{k \in \N_0}$. 
\end{itemize}

\end{lemma}

\begin{proof}

Since $\sum_{k \geq 0} |b_k| = c$ for some constant $c$, we proceed
$$ \big| \sum_{k \geq 0} \eps(1-\eps)^k b_k \big| \leq \eps \sum_{k \geq 0} | b_k | = \eps c \to 0$$
when $\eps \to 0$. And then $(i)$ is proved. 

On the other hand, since $\{a_k\}$ is periodic with period $q$, there are constants $c_j$, with $j = 0,..,q-1$ such that $a_{qk + j} = c_j$ for all $k \geq 0$. Then, for every fixed $j$ we have
$$\sum_{k \geq 0} \eps(1-\eps)^{qk + j} a_{qk + j} = c_j(1-\eps)^j\sum_{k \geq 0} \eps(1-\eps)^{qk} = \frac{c_j (1-\eps)^j \eps }{1 - (1-\eps)^q}.$$
Observe that the denominator in the right-hand side corresponds to a polynomial $S(\eps)$ of degree $q$. Additionally, one can readily confirm that $S(0) = 0$ and $S'(0) \neq 0$. Therefore, it follows that $S(\eps) = \eps R(\eps)$, where $R$ denotes a polynomial of degree $q-1$, and satisfies $R(0) \neq 0$. Then 
$$\frac{c_j (1-\eps)^j \eps }{1 - (1-\eps)^q} = \frac{c_j (1-\eps)^j }{R(\eps)} \to \frac{c_j}{R(0)}$$
when $\eps \to 0$. Therefore, we can infer that $\lim_{\eps \to 0} \sum_{k \geq 0} \eps(1-\eps)^k a_k = \sum_{j=0}^{q-1} \frac{c_j}{R(0)}$, establishing the proof of $(ii)$.

\end{proof}

Using the previous lemma, we can prove the main result of this section.

\begin{theorem}\label{teo:limit}
Let $N = \{1,...,n\}$ be a set of agents with delegation weightings $x_1,...,x_n$, and $P$ their delegation matrix. Then the limit in \eqref{eq:VP} is well defined. 
\end{theorem}

\begin{proof}
To understand the limiting behavior of $V^\eps$, the strategy involves expressing the matrix $(I - (1-\eps)\Pt)$ as a sum of powers of $\Pt$, thereby exploiting the structure of $\Pt$.

We have 
$$ I - \Pt^{\eps} = 	\begin{pmatrix} 
	I - (1-\eps)\Pt & \mathbf{0} \\
	-\bm{\varepsilon} & 1 \\
	\end{pmatrix}, $$
with $\Pt = P - diag \big( (P_{11},...,P_{nn} )\big).$
For the sake of simplicity we define $A_{\eps} := 	I - (1-\eps)\Pt$. Since $\eps > 0$ and $P$ is a stochastic matrix, we have $\|(1-\eps)\Pt\| < 1$. Thus $A_{\eps}^{-1}$ exists, and can be computed using the identity
\begin{equation}\label{potencia}
    A^{-1}_{\eps} = \sum_{k \geq 0} (1-\eps)^k\Pt^k.
\end{equation}
Utilizing the blockwise inversion formula yields:
$$ ( I - \Pt^{\eps} )^{-1} = 	\begin{pmatrix} 
	A_{\eps} & \mathbf{0} \\
	-\bm{\varepsilon} & 1 \\
	\end{pmatrix}^{-1} = 
	\begin{pmatrix} 
	A^{-1}_{\eps} & \mathbf{0} \\
	\bm{\varepsilon}A^{-1}_{\eps} & 1 \\
	\end{pmatrix}, $$
from we can derive
$$u^{\eps} = \begin{pmatrix} 
	A^{-1}_{\eps} & \mathbf{0} \\
	\bm{\varepsilon}A^{-1}_{\eps} & 1 \\
	\end{pmatrix}\bm{1}_0 = 
	\begin{pmatrix} A^{-1}_{\eps} \bm{1} \\ \bm{\eps} A^{-1}_{\eps} \bm{1} \end{pmatrix},$$
where $\bm{1} = (1,...,1) \in \R^n$.

Now we focus on the computation of $A^{-1}_{\eps} \bm{1} $. Using \eqref{potencia} we obtain $A^{-1}_{\eps} \bm{1} = \sum_{k \geq 0} (1-\eps)^k\Pt^k \bm{1}$. In order to analyze the powers of $\Pt$ we define:
\begin{align}
N_1 & := \{ i \in N \text{ such that } P_{ii} \not = 0 \}, \label{eq:N1}\\ 
N_2 & := \{ i \in N \setminus N_1 \text{ such that exists a delegation path from } i \text{ to } j \in N_1\}, \label{eq:N2} \\
N_3 & := N \setminus ( N_1 \cup N_2 ). \label{eq:N3}
\end{align}
Given that $N = N_1 \overset{d}{\cup} N_2 \overset{d}{\cup} N_3$, we adopt, without loss of generality, an indexing of the agents such that those belonging to $N_1$ appear first, those belonging to $N_2$ appear second, and those belonging to $N_3$ appear last. From this, we observe that $\Pt$ can be written as 
\begin{equation} \label{eq:numeric}
\Pt = \begin{pmatrix} 
	P_1 & \bm{0} \\
	P_2 & P_3 \\
\end{pmatrix}, 
\end{equation}
with $P_1$ being a sub-stochastic matrix, and $P_3$ a stochastic matrix. It is well known that the powers of $\Pt$ can be computed as
\begin{equation} \label{eq:block_power}
    \Pt^k = \begin{pmatrix} 
	P_1^k & \bm{0} \\
	B_k & P_3^k \\
\end{pmatrix},
\end{equation}
with $k \in \N$, and $B_k = \sum^{k-1}_{j=0}P_3^j P_2 P_1^{k - j -1}$. Then we have
$$
 \Pt^k \bm{1} = \begin{pmatrix} 
	P_1^k & \bm{0} \\
	B_k & P_3^k \\
\end{pmatrix}\bm{1} = \begin{pmatrix} 
	P_1^k\bm{1} \\
	B_k \bm{1} + P^k_3\bm{1} \\
\end{pmatrix},
$$
as a consequence
\begin{equation}\label{eq:numeric2}
    A^{-1}_{\eps} \bm{1} = \sum_{k \geq 0} (1-\eps)^k\Pt^k \bm{1} = \begin{pmatrix} 
	\sum_{k \geq 0}(1-\eps)^k P_1^k\bm{1} \\
		\sum_{k \geq 0}(1-\eps)^k (B_k \bm{1} + P^k_3\bm{1}) \\
\end{pmatrix}.
\end{equation}
Since $P_1$ is a sub-stochastic matrix whose columns indexed by $N_1$ add up to less than one, and every agent in $N_2$ has a delegation path to $N_1$, its spectral radius $\rho(P_1)<1$. Then, each entry of $P_1^k\bm{1}$ decays exponentially, and $\lim_{\eps \to 0} 	\sum_{k \geq 0}(1-\eps)^k P_1^k\bm{1}$ exists. Note that we do not need to compute $\sum_{k \geq 0}(1-\eps)^k (B_k \bm{1} + P^k_3\bm{1})$ since these entries correspond to agents $i \in N_3$ with $P_{ii} = 0$, and then $V^{\eps}_i(P) = 0$ for all $\eps >0$. 

Now our aim is to compute $\lim_{\eps \to 0}\bm{\eps} A^{-1}_{\eps} \bm{1}$. By exploiting the fact that $\bm{\eps} A^{-1}_{\eps} \bm{1} = \eps\, \bm{1}^{\top}\! A^{-1}_{\eps} \bm{1}$, our focus is solely on studying the behavior of $\eps A^{-1}_{\eps} \bm{1}$. Specifically, we obtain
$$\eps A^{-1}_{\eps} \bm{1} = \begin{pmatrix} 
	\sum_{k \geq 0}\eps(1-\eps)^k P_1^k\bm{1} \\
		\sum_{k \geq 0}\eps(1-\eps)^k (B_k \bm{1} + P^k_3\bm{1}) \\
\end{pmatrix}. $$
As previously established, $\lim_{\eps \to 0} \sum_{k \geq 0}(1-\eps)^k P_1^k\bm{1}$ exists. In this regard, we proceed to examine $\sum_{k \geq 0}\eps(1-\eps)^k (B_k \bm{1} + P^k_3\bm{1})$. Notably, 
the sequence $B_k \bm{1}$ need not converge (take $P_3$ a cycle of length two). However, writing $B_k \bm{1} = \sum_{m=0}^{k-1} P_3^{k-1-m} v_m$, with $v_m := P_2 P_1^m \bm{1}$ decaying exponentially, and using the decomposition $P_3^j = C_1 + C_2^j + C_3^j$ established below, each entry of $B_k \bm{1}$ is the sum of a periodic sequence and a summable one. Hence, from Lemma \ref{lemma:sequences}, $\lim_{\eps \to 0} \sum_{k \geq 0}\eps(1-\eps)^k B_k \bm{1}$ exists.

We proceed to compute the expression $\lim_{\eps \to 0} \sum_{k \geq 0}\eps(1-\eps)^k P^k_3 \bm{1}$. It is pertinent to note that the only information available regarding $P_3$ is that it is a stochastic matrix. The powers of $P_3$ can be represented as:
$$P^k_3 = C_1 + C^k_2 + C^k_3,$$
where $C_i = X J_i X^{-1}$, being $J$ the Jordan normal form $P_3 = XJX^{-1}$, with $J = J_1 + J_2 + J_3$. Here, $J_1$ consists solely of the Jordan blocks corresponding to the eigenvalue $\lambda = 1$, $J_2$ encompasses the Jordan blocks of the eigenvalues $\lambda$ satisfying $|\lambda|=1$ and $\lambda \not = 1$, and $J_3$ comprises the remaining blocks. Since $P_3$ is a stochastic matrix, every eigenvalue $\lambda$ with $|\lambda|=1$ is a root of unity and semisimple, i.e. its Jordan blocks are $1 \times 1$. Consequently, $C_1^k = C_1$ and the sequence $C_2^k \bm{1}$ is periodic. Further, as per the definition of $J_3$, it follows that the sequence $C_3^k \bm{1}$ converges exponentially to $\bm{0}$.

From this we have 
$$\sum_{k \geq 0}\eps(1-\eps)^k P^k_3 \bm{1} = \sum_{k \geq 0}\eps(1-\eps)^k C_1 \bm{1} + \sum_{k \geq 0}\eps(1-\eps)^k C^k_2 \bm{1} + \sum_{k \geq 0}\eps(1-\eps)^k C^k_3 \bm{1}, $$
$$ = C_1 \bm{1} + \sum_{k \geq 0}\eps(1-\eps)^k C^k_2 \bm{1} + \sum_{k \geq 0}\eps(1-\eps)^k C^k_3 \bm{1}.$$
Based on the findings in Lemma \ref{lemma:sequences}, it is concluded that the limit as $\eps$ tends to zero exists for both the second and third terms. Thus, the theorem is proven. (Alternatively, since Theorem \ref{teo:conservation} below holds for each $\eps > 0$ independently of this result, the limit of the last entry of $V^\eps$ follows from that of the first $n$ entries.)
\end{proof}

\subsection{Generalization}\label{sec:generalization}
Up to this point, we have made the assumption that each agent holds an equal role in the decision-making process. In the context of the particle system introduced in Section \ref{sec:particles}, this assumption is reflected in the distribution of one particle to each agent per unit time. This is formalized in the definition of $V$, where the vector $\bm{1}_0$ is used to construct $u^{\eps}$ in \eqref{eq:Ueps}. To modify the inherent voting weight of each agent, we need only adjust the definition of $u^{\eps}$ by replacing the vector $\bm{1}_0$ with an alternative vector $f \in \R^{ n+1}$, where $f_i$ represents the inherent voting weight of agent $i$. We can then generalize the definitions of $u^{\eps}$, $V^{\eps}(P)$, and $V(P)$ as follows: Let $N = \{1,...,n\}$ denote a set of agents with delegation weightings $x_1,...,x_n$, $P$ their delegation matrix, and $f \in \R^{n+1}$. We define
\begin{equation}\label{eq:Weps}
    u^{\eps} := (I - \Pt^{\eps})^{-1} f,
\end{equation}
\begin{equation}\label{eq:VPfeps}
    V^{\eps}(P,f) := (P^{\eps}_{11}u^{\eps}_1,...,P^{\eps}_{nn}u^{\eps}_n,u^{\eps}_{n+1}),
\end{equation}
and 
\begin{equation}\label{eq:VPf}
    V(P,f) := \lim_{\eps \to 0} V^{\eps}(P,f).
\end{equation}
In what follows we write $V^0 := V$. It is noteworthy that the proof of Theorem \ref{teo:limit} can be straightforwardly adapted to encompass the extended definitions. Therefore, we obtain:
\begin{theorem}\label{teo:limit_f}
Let $N = \{1,...,n\}$ be a set of agents with delegation weightings $x_1,...,x_n$, $P$ their delegation matrix, and $f \in \R^{n+1}$. Then the limit in definition \eqref{eq:VPf} is well defined. 
\end{theorem} 

\subsection{Self-selection }

The self-selection property P3 can be formalized as follows:

\begin{definition}[Self-selection property]
Let $N = \{1, \ldots, n\}$ be a set of agents with delegation weightings $x_1, \ldots, x_n \in X$, $P$ their delegation matrix, and $f \in \R^{n+1}$ with $f_i \geq 0$ for all $i \in N$. We say that $V(P, f)$ satisfies the self-selection property if:
    \begin{itemize}
        \item $V_i(P,f) \geq P_{ii}f_i$, for all $i \in N$,
        \item $P_{ii} = 0$ implies $V_i(P,f) = 0$, for all $i \in N$.
    \end{itemize}
\end{definition}

The following lemma asserts that a non-negative value of the term $f$ implies a non-negative voting weight for each agent.

\begin{lemma} \label{Lem:upositive}
    Let $N = \{1,...,n\}$ be a set of agents with delegation weightings $x_1,...,x_n \in X$, $P$ their delegation matrix, and $f \in \R^{n+1}$ with $f_i \geq 0$ for all $i \in N$. Then $V_i^{\eps}(P,f) \geq 0$ for all $i \in N$, for all $\eps \in [0,1)$. 
\end{lemma}

\begin{proof}
    Given $\eps \in (0,1)$, consider $u^{\eps}$ as in \eqref{eq:Weps}. Since the last column of $\Pt^\eps$ vanishes, the first $n$ entries of $u^\eps$ do not depend on $f_{n+1}$, so we may assume $f_{n+1} = 0$ and hence $f \geq 0$. As $\rho(\Pt^\eps) < 1$ (see Section \ref{sec:formal}) and $\Pt^\eps$ has nonnegative entries, $u^\eps = \sum_{m \geq 0} (\Pt^\eps)^m f \geq 0$. 
 Henceforth, it is established that $u_i^\eps \geq 0$ for all $i \in N$. As a consequence, it can be deduced that $V_i^{\eps}(P,f) \geq 0$ holds true for all $i \in N$ and for all $\eps \in (0,1)$. Furthermore, the scenario $\eps = 0$ can be trivially derived by taking the limit as $\eps$ approaches 0.
\end{proof}

The application of the aforementioned result allows us to establish the generalized version of the self-selection property (P3). The subsequent theorem states that an agent's voting weight is zero if they choose to delegate their entire voting weight, whereas if they retain a fraction of their voting weight for themselves, then their voting weight is no less than their inherent voting weight scaled by this fraction. The theorem is stated as follows.

\begin{theorem}\label{teo:propP3}
    Let $N = \{1,...,n\}$ be a set of agents with delegation weightings $x_1,...,x_n \in X$, $P$ their delegation matrix, and $f \in \R^{n+1}$ with $f_i \geq 0$ for all $i \in N$. Then we have:
    \begin{itemize}
        \item $P_{ii} = 0$ implies $V_i^\eps(P,f) = 0$,
        \item $V_i^\eps(P,f) \geq (1-\eps)P_{ii}f_i$ for all $i \in N$, and $\eps \in [0,1)$.
    \end{itemize}
\end{theorem}
\begin{proof}
    Given that $V_i^\eps(P,f) = P^{\eps}_{ii}u_i^\eps$, it can be readily verified that $P_{ii} = 0$ implies $V_i^\eps(P,f) = 0$ for all $1 \leq i \leq n$ and all $\eps \in (0,1)$. Moreover, from Lemma \ref{Lem:upositive}, which establishes that $u^\eps_i \geq 0$ for all $i \in N$, since $f_i \geq 0$ for all $i \in N$, we can estimate using \eqref{eq:Weps} that:
    $$
    u^\eps_i = f_i + \Pt^\eps_{i1}u_1 +... + \Pt^\eps_{in}u_n \geq f_i.
    $$
As a result, we obtain $V_i^\eps(P,f) = P^{\eps}_{ii}u^\eps_i \geq (1-\eps)P_{ii}f_i$ for all $i \in N$, from which we can deduce the desired outcome for all $\eps \in (0,1)$. The scenario $\eps = 0$ can be easily handled by means of the aforementioned derivation, as we take the limit $\eps \to 0$.
\end{proof}

\subsection{Conservation of total voting weight}

The preservation of voting weight is a desired property in the generalization of the classic LD model, which requires that the sum of the voting weight of each agent and the lost voting weight in delegation cycles equals the number of agents. Although this property is usually achieved through the normalization of the voting weight vector, it is demonstrated in our case that this step is unnecessary and, therefore, preserves the probabilistic interpretation of the proposed voting measure discussed in Section \ref{sec:preli}. The proof of this property is presented in the following result, and its proof is given in Appendix \ref{proof:teo_cons}.

\begin{theorem}\label{teo:conservation}
Let $N = \{1,...,n\}$ be a set of agents with delegation weightings $x_1,...,x_n \in X$, $P$ their delegation matrix, and $f \in \R^{n+1}$. Then $\sum^{n+1}_{i=1} V^{\eps}_i(P,f) = \sum^{n+1}_{i=1} f_i$ for all $\eps \in [0,1)$. 
\end{theorem}

\subsection{Numerical computation}\label{sec:numerical}

From the proof of Theorem \ref{teo:limit}, it is possible to obtain a procedure for computing $V$ that avoids the need for computing $V^{\eps}$. This procedure is used later for technical purposes. Specifically, \eqref{eq:numeric} and \eqref{eq:numeric2} imply that only the inverse of $(I - P_1)$ needs to be computed, where $P_1$ is a sub-stochastic matrix. This is due to the fact that the voting weight of agents in $N_3$ (and $N_2$) is zero, and therefore, it suffices to obtain the voting weight of agents in $N_1$. This observation is summarized in Algorithm \ref{alg:proc}.
\begin{algorithm}[ht]
\caption{\textbf{Computation of $V(P,f)$}}
\label{alg:proc}
\begin{flushleft}
\textbf{Input:} A set of agents $N = \{1,...,n\}$, their delegation weightings $x_1,...,x_n$, and a vector $f \in \R^{n+1}$;
\\
$\boldsymbol{1}$: Define the delegation matrix $P$ using the weightings $x_1,...,x_n$, and $N_1$ and $N_2$ as in \eqref{eq:N1} and \eqref{eq:N2} respectively. Define $I_i$ as the $i$-th element of $N_1 \cup N_2$, with $1 \leq i \leq \#(N_1 \cup N_2)$;
\\
$\boldsymbol{2}$: Define $\Pt := P - diag( (P_{11},...,P_{nn}) )$, and $\Pt_r$ as the restriction of $\Pt$ to $N_1 \cup N_2$. That is, remove from $\Pt$ the $i$-th row and the $i$-th column, for all $i \not \in N_1 \cup N_2$. Analogously, define $f_r$ removing the $i$-th entry of $f$ for all $i \not \in N_1 \cup N_2$. 
\\
$\boldsymbol{3}$: If $N_1 \cup N_2$ $\not = \emptyset$, compute $u_r = (I - \Pt_r)^{-1} f_r$;
\\
$\boldsymbol{4}$: Compute $V_{I_i}(P,f) = P_{I_i I_i} [u_r]_i$ for all $1 \leq i \leq \#(N_1 \cup N_2)$, and $V_i(P,f) = 0$ for all $i \not \in N_1 \cup N_2 \cup \{n+1\}$. Finally, using Theorem \ref{teo:conservation}, compute $V_{n+1}(P,f) = \sum^{n+1}_{i=1} f_i - \sum^n_{i=1} V_i(P,f)$;
\end{flushleft}
\end{algorithm}

\subsection{Alternative definitions of $V$}\label{sec:alternative} 

As we mentioned before, $V$ can be defined directly without relying on the functions $V^\eps$. Algorithm \ref{alg:proc} provides a way to compute $V$ directly. It is worth noting that this procedure removes agents belonging to delegation cycles from the system to ensure that the matrix $I-P_1$ is invertible. However, considering that entries of $V$ that correspond to agents belonging to cycles are be multiplied by zero in the last step, $V$ can be defined in terms of powers of the delegation matrix. Specifically, let $P \in \calP^n$ and $f\in \R^n$ (the first $n$ entries of the source term), then $V: \calP^n \times \R^n \to \R^n$ can be defined as: 
\begin{equation}\label{eq:alter_def}
V(P,f) := \lim_{k \to \infty} \big(\sum^k_{\ell = 0} \Pt^{\ell} f\big) \odot diag(P), 
\end{equation}
where $\Pt = P - diag((P_{11},...,P_{nn}))$. It should be noted that the value associated with node $n+1$ may be computed utilizing Theorem \ref{teo:conservation}, as stipulated in step 4 of Algorithm \ref{alg:proc}.

On the other hand, an equivalent definition can be obtained using Markov Chains by considering two types of states for each agent: an absorbing state that accounts for their voting weight and a delegation state that manages interaction with the rest of the system. More precisely, for any voter not in a delegation cycle, we introduce two states: one "delegation state" and one "casting state". Additionally, all voters within a delegation cycle are contracted into a single "null state". For any non-cycle agent, we add a transition probability of $x_{ii}$ from its delegation state to its casting state. For any other agent $j$, if $j$ is a non-cycle agent, we add a transition probability of $x_{ij}$ from $i$'s delegation state to $j$'s delegation state. If $j$ is a cycle-agent, we add a transition probability of $x_{ij}$ to the null state. Thus, the resulting Markov chain is absorbing, with absorbing states corresponding precisely to the casting states and the null state. Utilizing the classical properties of Markov chains, a direct proof can be constructed to demonstrate that this measure satisfies properties (P1) to (P3).

\subsection{The generalization property}
The property (P1) can be formalized as follows:

\begin{definition}[Generalization property]\label{def:P1}
    Let $\calB^n := \{ B \in \calP^n : B_{ii} \in \{0,1\} \text{ for all } 1 \leq i \leq n \}$, let $V: \calP^n \to \R^{n+1}$, and $V^c$ the classic voting measure defined in \eqref{eq:classicLD} (extended to $\calP^n$, see \eqref{eq:classicLDG}). We say that $V$ has the generalization property if $V(P) = V^c(P)$ for all $P \in \calB^n$.
\end{definition}

Expanding upon the previous section's alternative definition of $V$, we establish that this voting measure is a generalization of the classic LD model, as it satisfies the property (P1). Moreover, the following result demonstrates that the proposed measure $V$ and the standard generalization of the classic model defined in \eqref{eq:classicLD} are equivalent, provided that agents are unable to retain a portion $q \in (0,1)$ of their vote. Its proof is given in Appendix \ref{proof:teo_propP1}.

\begin{theorem}\label{teo:propP1}
Let $\calB^n := \{ B \in \calP^n : B_{ii} \in \{0,1\} \text{ for all } 1 \leq i \leq n \}$, and let $V$ and $V^c$ the voting measures defined in \eqref{eq:VP}, and \eqref{eq:classicLD} respectively. Then $V_i(P) = V_i^c(P)$ for all $P \in \calB^n$, and $i \in N$. 
\end{theorem}

Analogous results hold for $V^\eps$, which establish that this voting weight measure can be interpreted as a generalization of the classic LD model with a specific margin of error, subject to the penalty parameter $\eps$. As before, its proof is provided in Appendix \ref{proof_teo_P1delta}. 

\begin{theorem}\label{teo:propP1delta}
Let $\calA^n := \{ A \in \calP^n : A_{ij} \in \{0,1\} \text{ for all } 1 \leq i,j \leq n \}$, and let $V^\eps$ and $V^c$ the voting measures defined in \eqref{eq:VPeps}, and \eqref{eq:classicLD} respectively. Given $\delta >0$ there exist $\eps_\delta >0$ such that if $\eps < \eps_\delta$ then $\max_{1\leq i \leq n}|V^\eps_i(P) - V_i^c(P)| < \delta$ for all $P \in \calA^n$. 
\end{theorem}

\subsection{The delegation property}

This section presents the formalization and proof of the delegation property (P2), which establishes the legitimacy of delegation. Specifically, the property posits that if agent $i$ delegates her voting weight to agent $j$, who subsequently delegates all of her voting weight to other agents, the resulting weight distribution would be equivalent to that obtained if agent $i$ had utilized delegation weighting $x_j$ instead of $x_i$. Furthermore, when agent $i$ assigns, let say, $\frac{1}{4}$ and $\frac{3}{4}$ of her voting weight to agents $j$ and $k$, respectively, and these proxies also delegate all their voting weight, the corresponding voting weight distribution would match that produced by the delegation weighting $\frac{1}{4}x_j + \frac{3}{4}x_k$ instead of $x_i$. Our examples assume that agents $j$ and $k$ do not delegate a portion of their voting weight back to agent $i$. However, in such cases, the fraction of the voting weight is distributed proportionally among the proxies. We formally state (P2) in the following definition:

\begin{definition}[Delegation property]\label{def:P2}
     Consider a set $N = \{1,...,n\}$ of agents with delegation weightings $x_1,...,x_n \in X$, their delegation matrix $P$, and $f \in \R^{n+1}$. Let $D \subset N$, such that $x_{ii} = 0$ for all $i \in D$, and suppose that there exists an agent $k \in N \setminus D$ such that $x_{kj} > 0$ for all $j \in D$, and $x_{kj}=0$ otherwise. Moreover, assume that agent $k$ is not involved in a delegation cycle. For each $i \in D$, we define $x^*_i \in \R^n$ as $x^*_{ij} := x_{ij}$ for all $j \not = k$ and $x^*_{ik} := 0$, and $x^*_k \in X$ as
\begin{equation}\label{eq:delegation}
    x^*_k := \frac{\sum_{i \in D} x_{ki}x^*_{i}}{ 1 - \sum_{i \in D} x_{ki}x_{ik} };
\end{equation}
We also define $P^* \in \R^{n \times n}$ as $P^* := (x_1 |... | x_{k-1} | x^*_k |x_{k+1}|... | x_n)$, which is the delegation matrix obtained by replacing $x_k$ with $x^*_k$ in $P$. Let $V: \calP^n \times \R^{n+1} \to \R^{n+1}$ a voting weight measure. We say that $V$ satisfies the delegation property if $V(P, f) = V(P^*, f)$ for all $P \in \calP^n$ and $f \in \R^{n+1}$.
\end{definition}

The following result establishes that the voting measure defined in Section \ref{sec:formal} satisfies the delegation property. \footnote{While this property can be readily shown using an alternative definitions of $V$ given in Section \ref{sec:alternative}, the techniques employed in the proof of Theorem \ref{teo:delegation} are subsequently applied to establish the $\delta$-Delegation property in Theorem \ref{teo:delta_delegation}, which is the main result of this section.
}

\begin{theorem}\label{teo:delegation}
  Consider a set $N = \{1, \ldots, n\}$ of agents with delegation weightings $x_1, \ldots, x_n \in X$, their delegation matrix $P$, and $f \in \R^{n+1}$. Let $V$ be the voting measure defined in Section \ref{sec:generalization}. Then, $V$ satisfies the delegation property given in Definition \ref{def:P2}.
\end{theorem}

\begin{proof}
To establish that $V(P,f) = V(P^*,f)$, the strategy involves noting that $V(P^*,f)$ can be represented in terms of $P$ through a change variable. This change is obtained by expressing a reduced version of the matrix $(I - \Pt^*)$ in terms of $(I - \Pt)$, achieved by operating on its columns.

 We begin by verifying the well-definedness of \eqref{eq:delegation}, that is, $1 - \sum_{i \in D} x_{ki}x_{ik} > 0$. To establish this, we assume the contrary, i.e., $\sum_{i \in D} x_{ki}x_{ik} = 1$. Since $0 \leq x_{ik}, x_{ki} \leq 1$ for all $i \in N$, and $\sum_{i \in D} x_{ki} = 1$, it follows that $x_{ik} = 1$ for all $i \in D$. Subsequently, since $x_{kj} = 0$ for all $j \not \in D$, it implies that $D \cup \{ k \}$ forms a delegation cycle, which contradicts our assumption that agent $k$ is not part of a cycle.

 To apply the procedure detailed in Section \ref{sec:numerical}, we introduce the sets $N_1$, $N_2$, and $N_3$ as specified in \eqref{eq:N1}, \eqref{eq:N2}, and \eqref{eq:N3}, respectively. Similarly, we define $N^*_1$, $N^*_2$, and $N^*_3$ using the delegation weighting $x^*_k$ instead of $x_k$. We assume that the agents in $N$ are arranged in a way that agents in $N_1$ appear first, followed by agents in $N_2$, with $k$ in the last place so that $k = \#(N_1 \cup N_2)$, and then agents in $N_3$, for simplicity. Note that an agent $i \in N$ belongs to a cycle if and only if $i>k$. Let us define $D' := D \cap (N_1 \cup N_2)$, and observe that $D'$ cannot be empty, otherwise $k$ would belong to $N_3$, which contradicts our assumption.

We now establish that $N^*_3 = N_3$. Firstly, since agent $k$ is not part of any cycle, changing her delegation weighting by $x^*_k$ cannot resolve any cycle. Therefore, $N_3 \subseteq N^*_3$. On the other hand, since $k \not\in N_3$, a shortest delegation path from $k$ to $N_1$ does not revisit $k$, so its first two agents $q \in D'$ and $s \in (N_1 \cup N_2) \setminus \{k\}$ satisfy $x^*_{qs} = x_{qs} > 0$. By definition \eqref{eq:delegation}, it follows that $x^*_{ks} > 0$. Thus, agent $k$ does not belong to $N^*_3$. Consequently, no new delegation cycle is formed when we replace $x_k$ with $x^*_k$, and we can deduce that $N^*_3 \subseteq N_3$.

The equality $N_3 = N^*_3$ implies that $N_1 = N^*_1$ and $N_2 = N^*_2$. Consequently, Algorithm \ref{alg:proc} can be applied to $P$ and $P^*$, and it can be observed that their restricted matrices, denoted by $\Pt_r$ and $\Pt^*_r$ respectively, share the same dimension. Specifically, we have $\Pt_r, \Pt^*_r \in \R^{k \times k}$. To determine $V(P,f)$ and $V(P^*,f)$, we need to solve the linear systems:
\begin{equation*}
    (I - \Pt_r)u = f_r,
\end{equation*}
and 
\begin{equation*}
    (I - \Pt^*_r)u^* = f_r
\end{equation*}
respectively, where $f_r$ is the restricted source term defined in Algorithm \ref{alg:proc}. It suffices to show that $u_i = u^*_i$ holds for all $i \in N_1$, given that $V_i(P,f) = V_i(P^*,f) = 0$ for all $i \in N_2 \cup N_3$, and that the last entries then coincide by Theorem \ref{teo:conservation}. To this end, we introduce the matrix $A := (I - \Pt_r)$ and we define the variable $u'$ as follows:
\begin{equation}\label{eq:changev}
\left\lbrace
\begin{aligned}
u'_i & = u_i - x_{ki}u_k &\mbox{ for all } i \in D', \\
u'_k & = u_k (1 - \sum_{i \in D'} x_{ki}x_{ik}), \\
u'_i & = u_i &\mbox{ for all } i \not \in D' \cup \{k\}.
\end{aligned}
\right.
\end{equation}
From this, it can be seen that $u' \in \R^k$ solves the linear system 
\begin{equation*}
    \Big(A_1 \Big|...\Big|A_{k-1}\Big| \frac{A_k+\sum_{i \in D'} x_{ki}A_i}{(1 - \sum_{i \in D'} x_{ki}x_{ik})}\Big)u' = f_r,
\end{equation*}
where $A_i$ denotes the $i$-th column of $A$ and $A_{ij}$ its $j$-th entry. Now, we observe that $A_i = I_i - x_i$ for all $i \in D'$, then 
\begin{equation}\label{eq:exp_delta1}
    A_{ki} + x_{ki}A_{ii} = -x_{ki} + x_{ki} - x_{ki}x_{ii} = 0,
\end{equation}
and
\begin{equation}\label{eq:exp_delta2}
    A_{kk} + \sum_{i \in D'} x_{ki}A_{ik} = 1 - \sum_{i \in D'} x_{ki}x_{ik}.
\end{equation}
Then we have
\begin{equation}\label{eq:exp_delta3}
    A_k+\sum_{i \in D'} x_{ki}A_i = (1 - \sum_{i \in D'} x_{ki}x_{ik})I_k - \sum_{i \in D'} x_{ki}x^*_{i}.
\end{equation}
Based on the dependence of $\Pt^*_r$ construction solely on $x^*_{ki}$ values for $i\leq k$, that are determined by $x_i$ for $i \in D' \cup {k}$, we can infer that
\begin{equation*}
    (I - \Pt^*_r) = \Big(A_1 \Big|...\Big|A_{k-1}\Big| \frac{A_k+\sum_{i \in D'} x_{ki}A_i}{(1 - \sum_{i \in D'} x_{ki}x_{ik})}\Big).
\end{equation*}
In conclusion, we establish the equality $u' = u^*$ and, by virtue of the expression given in \eqref{eq:changev}, deduce that $u_i = u'_i = u^*_i$ holds for every $i \in N_1$. Thus, the theorem is demonstrated.
\end{proof}

To provide further clarification regarding Theorem \ref{teo:delegation} and its scope, let us first consider the example depicted in Figure \ref{fig:example0}. This example showcases a delegation scenario and its equivalent representation as per Theorem \ref{teo:delegation}, using $N=\{1,...,4\}$, $k=3$, $D = \{2\}$. Obtaining the equivalent delegation setting implies replacing $x_3$ with $x^*_3$ based on \eqref{eq:delegation}. The resulting configuration reveals that delegating all the voting weight to agent $2$ is equivalent to replicating her delegation weighting.

The example shown in Figure \ref{fig:example1} demonstrates the outcome when an agent delegates to another agent who, in turn, delegates a fraction of her voting weight back to the first agent. The delegation weighting in this case follows the same parameters as the example in Figure \ref{fig:example0}, where $N=\{1,...,4\}$, $k=3$, and $D=\{2\}$, according to Theorem \ref{teo:delegation}. Using the delegation formula in \eqref{eq:delegation}, we find that if agent $3$ delegates all her voting weight to agent $2$, this is equivalent to copying agent $2$'s delegation weighting, disregarding the fraction delegated by $2$ to $3$, and normalizing the resulting vector to obtain a delegation weighting. Specifically, given $x_2=(\frac{1}{2},0,\frac{1}{4},\frac{1}{4})$ and $x_3=(0,1,0,0)$, we have $x^*_3=(\frac{1}{2},0,0,\frac{1}{4})/(1-\frac{1}{4})=(\frac{2}{3},0,0,\frac{1}{3})$.

The verification of equivalence in both examples can be easily achieved by applying Algorithm \ref{alg:proc}. Additionally, Theorem \ref{teo:delegation} establishes that delegating an agent's entire voting weight to multiple agents who, in turn, delegate their entire voting weight further, is equivalent to changing the delegation weighting through a convex combination of their profiles. This convex combination is obtained by utilizing either the modified profiles as shown in Figure \ref{fig:example1}, or the unmodified profiles as illustrated in Figure \ref{fig:example0}, as stipulated by equation \eqref{eq:delegation}.

\begin{figure}[ht]
\begin{center}
	\begin{tikzpicture}[->, >=stealth', auto, semithick, node distance=1cm]
	\tikzstyle{every state}=[fill=white,draw=black,thick,text=black,scale=0.7]
	\node[state] (A) at (0,2) {$1$};
	\node[state] (B) at (1,0.5) {$2$};
	\node[state] (C) at (2,2) {$3$};
	\node[state] (C2) at (2.5,0.5) {$4$};
        \node[state] (D) at (4,2) {$1$};
	\node[state] (E) at (5,0.5) {$2$};
	\node[state] (F) at (6,2) {$3$};
	\node[state] (F2) at (6.5,0.5) {$4$};
	\path
	(A) edge[loop left]			node{$1$}	(A)
	(B) edge[below]	node{$\frac{1}{2}$}	(A)
	edge[left,below]	node{$\frac{1}{2}$}	(C2)
	(C) edge[left,below]	node{$1$}	(B)
	(C2) edge[loop right]	node{$1$}	(C2)
	(D) edge[loop left]			node{$1$}	(D)
	(E) edge[below]	node{$\frac{1}{2}$}	(D)
	    edge[left,below]		node{$\frac{1}{2}$}	(F2)
	(F) edge[above]	node{$\frac{1}{2}$}	(D)
	    edge[above]	node{$\frac{1}{2}$}	(F2)
	(F2) edge[loop right]	node{$1$}	(F2);
	%\node[above=0.5cm] (A){Patch G};
	%\draw[red] ($(D)+(-1.5,0)$) ellipse (2cm and 3.5cm)node[yshift=3cm]{Patch H};
	\end{tikzpicture}
\end{center}
\caption{According to Theorem \ref{teo:delegation}, the delegation configurations depicted above are equivalent. It can be observed that the act of delegating all the voting weight to agent $2$ is equivalent to replicating her delegation weighting.}
\label{fig:example0}
\end{figure}
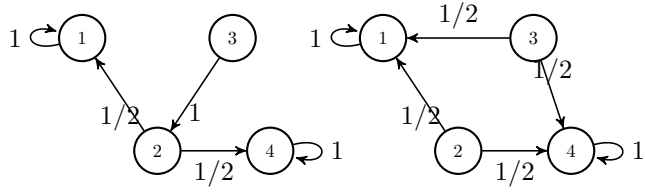

\begin{figure}[ht]
\begin{center}
	\begin{tikzpicture}[->, >=stealth', auto, semithick, node distance=1cm]
	\tikzstyle{every state}=[fill=white,draw=black,thick,text=black,scale=0.7]
	\node[state] (A) at (0,2) {$1$};
	\node[state] (B) at (1,0.5) {$2$};
	\node[state] (C) at (2,2) {$3$};
	\node[state] (C2) at (2.5,0.5) {$4$};
    \node[state] (D) at (4,2) {$1$};
	\node[state] (E) at (5,0.5) {$2$};
	\node[state] (F) at (6,2) {$3$};
	\node[state] (F2) at (6.5,0.5) {$4$};
	\path
	(A) edge[loop left]			node{$1$}	(A)
	(B) edge[below]	node{$\frac{1}{2}$}	(A)
	edge[bend left,above]	node{$\frac{1}{4}$}	(C)
	edge[left,below]	node{$\frac{1}{4}$}	(C2)
	(C) edge[bend left,below]	node{$1$}	(B)
	(C2) edge[loop right]	node{$1$}	(C2)
	(D) edge[loop left]			node{$1$}	(D)
	(E) edge[below]	node{$\frac{1}{2}$}	(D)
	    edge[left,below]		node{$\frac{1}{4}$}	(F)
	    edge[left,below]		node{$\frac{1}{4}$}	(F2)
	(F) edge[above]	node{$\frac{2}{3}$}	(D)
	    edge[above]	node{$\frac{1}{3}$}	(F2)
	(F2) edge[loop right]	node{$1$}	(F2);
	%\node[above=0.5cm] (A){Patch G};
	%\draw[red] ($(D)+(-1.5,0)$) ellipse (2cm and 3.5cm)node[yshift=3cm]{Patch H};
	\end{tikzpicture}
\end{center}
\caption{According to Theorem \ref{teo:delegation}, the delegation setting on the left is equivalent to the one on the right. That is, when agent $3$ delegates the voting weight to agent $2$, the effect is the same as copying her delegation weighting and redistributing the voting weight that agent $2$ delegates back to $3$ among her proxies in proportion to their respective weights.}
\label{fig:example1}
\end{figure}
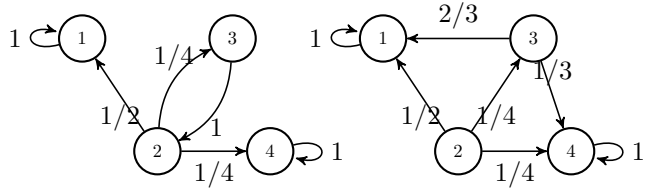

The definition of the voting measure $V^\eps$ includes a penalty on the length of delegation chains, which causes these measures to violate the delegation property. We observe that replicating an agent's delegation weighting results in a configuration with shorter delegation chains, which in turn applies less penalty and produces a different distribution of voting weights. Nevertheless, we demonstrate that $V^\eps$ conforms to the delegation property within a certain margin of error that relies on the penalty $\eps$. Specifically, a reduction in the penalty leads to a decrease in the error. We establish and verify this property, denoted as the $\delta-$Delegation property, in the following theorem.
\begin{theorem}[$\delta-$Delegation property]\label{teo:delta_delegation}
  Assuming the identical hypotheses as stipulated in Theorem \ref{teo:delegation}, and $f \in \R^{n + 1}_{\geq 0}$. Given $\delta>0$, there exists a positive $\eps_{\delta}$ such that if $\eps<\eps_\delta$ then $\|V^\eps(P,f) - V^\eps(P^*,f)\|_{\infty} < \delta$.
\end{theorem}
\begin{proof}
Define
$$A := (I-\Pt^\eps),$$
$$A^* := (I-\Pt^{*\eps}),$$
$$u := A^{-1}f,$$
and
$$u^* := A^{*-1}f.$$
  Operating with the columns of the matrix $A$ and making the appropriate change of variables as in the proof of Theorem \ref{teo:delegation}, we introduce a new variable, denoted by $u' \in \R^{n+1}$, which is defined as follows:
\begin{equation}\label{eq:changev_delta}
\left\lbrace
\begin{aligned}
u'_i & = u_i - (1-\eps)x_{ki}u_k &\mbox{ for all } i \in D, \\
u'_k & = u_k (1 - \sum_{i \in D} (1-\eps)^2x_{ki}x_{ik}), \\
u'_i & = u_i &\mbox{ for all } i \not \in D \cup \{k\},
\end{aligned}
\right.
\end{equation}
and fulfills the equation $A'u' = f$, with
\begin{equation*}
     A' := \Big(A_1 \Big|...\Big|A_{k-1}\Big| \frac{A_k+\sum_{i \in D} (1-\eps)x_{ki}A_i}{(1 - \sum_{i \in D} (1-\eps)^2x_{ki}x_{ik})} \Big| A_{k+1} \Big|...\Big| A_{n+1} \Big).
\end{equation*}
Let us define $\Delta := A' - A^*$. From this definition, it follows that $A^* u' = f - \Delta u'$. By combining this expression with the definition of $u^*$, we arrive at:
\begin{equation}\label{eq:deltadel1}
    A^*(u^* - u') = \Delta u'.
\end{equation}
 By invoking the definition of $A^*$ and following a similar approach as carried out in equations \eqref{eq:exp_delta1}, \eqref{eq:exp_delta2}, and \eqref{eq:exp_delta3}, and denoting $\Delta_i$ as the $i$-th column of $\Delta$, it can be established that $\Delta_i = \bm{0}$ for all $i \not = k$, $\Delta_{kk} = 0$, and
\begin{equation}\label{eq:Delta_column}
\begin{aligned}
    \Delta_{kj} &= (1-\eps)\frac{\sum_{i \in D} x_{ki}x^*_{ij}}{ 1 - \sum_{i \in D} x_{ki}x_{ik} } - (1-\eps)^2 \frac{\sum_{i \in D} x_{ki}x^*_{ij}}{ 1 - (1-\eps)^2 \sum_{i \in D} x_{ki}x_{ik} } \quad \text{for } 1 \leq j \leq n,\ j \neq k, \\
    \Delta_{k\,n+1} &= - \sum_{j=1}^{n} \Delta_{kj},
\end{aligned}
\end{equation}
where the last identity holds because the columns of $A$, $A^*$ and $A'$ indexed by $D \cup \{k\}$ add up to zero. Hence $\Delta u' = u'_k\Delta_k$. By Lemma \ref{Lem:upositive} and \eqref{eq:changev_delta}, $0 \leq u'_k \leq u_k$. Moreover, since $k \in N_1 \cup N_2$ and $f \geq 0$, from \eqref{potencia} and the block structure \eqref{eq:numeric} we get $u_k = \sum_{m \geq 0}(1-\eps)^m [P_1^m f_r]_k \leq [(I - P_1)^{-1} f_r]_k =: u_k(P,f)$, the entry of $u_r$ corresponding to $k$ in Algorithm \ref{alg:proc}. Hence
\begin{equation}\label{eq:deltadel2}
    0 \leq u'_k \leq u_k(P,f).
\end{equation}
On the other hand, writing $s := \sum_{i \in D} x_{ki}x_{ik}$ and noting that $1 - (1-\eps)^2 s \geq 1 - s$ and that $\sum_{j \leq n}\sum_{i \in D} x_{ki}x^*_{ij} = 1 - s$, it is possible to verify from \eqref{eq:Delta_column} that $\Delta_{kj} \geq 0$ for all $1 \leq j \leq n$ and
\begin{equation}\label{eq:deltadel3}
    0 \leq \sum_{j=1}^{n} \Delta_{kj} = (1-\eps) - \frac{(1-\eps)^2 (1-s)}{1 - (1-\eps)^2 s} = \frac{\eps(1-\eps)\big(1 + (1-\eps)s\big)}{1 - (1-\eps)^2 s} \leq \frac{2\eps}{1-s}.
\end{equation}
Now, since the last column of $A^*$ is $\bm{\delta}_{n+1}$, the first $n$ entries of $u^* - u'$ coincide with those of $A^{*-1} g$, where $g := u'_k(\Delta_{k1},\dots,\Delta_{kn},0) \geq 0$. Applying Lemma \ref{Lem:upositive} and Theorem \ref{teo:conservation} to the system $A^* v = g$, and using the estimates \eqref{eq:deltadel2} and \eqref{eq:deltadel3}, we obtain for every $i \in N_1$:
$$0 \leq P^{\eps}_{ii}(u^* - u')_i \leq \sum_{1\leq j \leq n} u'_k\Delta_{kj} \leq \frac{2\eps\, u_k(P,f)}{1-s}.$$
Since $u'_i = u_i$ for $i \in N_1$ by \eqref{eq:changev_delta}, this bounds $|V^\eps_i(P,f) - V^\eps_i(P^*,f)|$ for $i \in N_1$; for $i \in N_2 \cup N_3$ both quantities vanish, and for the last entry Theorem \ref{teo:conservation} yields a bound $n$ times larger. Defining the constant
\begin{equation}\label{eq:deltadel_constant}
    C(P,f,n,k) := \frac{2n\, u_k(P,f)}{ 1 - \sum_{j \in D} x_{kj}x_{jk}},
\end{equation}
   one can establish that $\|V^\eps(P,f) - V^\eps(P^*,f)\|_\infty \leq C(P,f,n,k)\eps$. Taking $\eps_\delta = \delta/C$ we conclude the proof. 
\end{proof}

\begin{remark}
The statement of Theorem \ref{teo:delta_delegation} also follows directly from Theorems \ref{teo:limit_f} and \ref{teo:delegation} by the triangle inequality,
$$\|V^\eps(P,f) - V^\eps(P^*,f)\|_\infty \leq \|V^\eps(P,f) - V(P,f)\|_\infty + \|V(P^*,f) - V^\eps(P^*,f)\|_\infty \to 0 \quad \text{as } \eps \to 0.$$
The advantage of the proof given above is that it provides the explicit constant \eqref{eq:deltadel_constant}, and hence a rate linear in $\eps$, which is relevant for studying the behavior of the system as $\eps \to 0$.
\end{remark}

\begin{remark}
Assuming fixed parameters $n$ and $f \in \R^{n+1}$, the dependence of $\eps_\delta$ on $P$ and $k$ is apparent from \eqref{eq:deltadel_constant}. However, the dependence on $k$ can be eliminated as follows: given $P \in \calP^n$, let $K_{P} \subset N$ be the set of agents that delegate all their voting weight and do not belong to a delegation cycle (that is, satisfying the conditions required for agent $k$ as stated in Definition \ref{def:P2}). We define
$$C'(P) := \max_{k \in K_P} C(P,f,n,k) = 2n \max_{k \in K_P} \frac{u_k(P,f)}{1 - \sum_{j \in D_k} x_{kj}x_{jk}},$$
and set $\eps_\delta = \delta/C'$. On the other hand, the dependence on $P$ cannot be easily disregarded since the constant $C'$ deteriorates as $\min_{k \in K_P} 1 - \sum_{j \in D_k} x_{kj}x_{jk}$ approaches zero or as $u_k(P,f)$ grows. However, it can be deduced that, for a fixed value of $\delta$ and $\eps$, $\|V^\eps(P,f) - V^\eps(P^*,f)\|_{\infty} < \delta$ holds for all $P \in \{ P \in \calP^n \text{ such that } C'(P) < \delta/\eps \}$. 
\end{remark}

\section{Equilibrium Analysis}\label{sec:nash}

\subsection{The Delegation Game}
In this paper, the decision-making process is examined as a continuous dynamic over time, rather than as a single event, as previously stated. The actions taken by agents, whether it be delegating voting weight or making use of it, have an impact on other agents. As a result, feedback is received by agents, who then adjust their delegation weightings and actions accordingly. This raises the question of whether or not equilibrium states exist for this type of mechanics.

To partially model this dynamic, we introduce the game $\mathcal{G}(N,U,W,\eps)$, wherein $N = \{1,...,n\}$ represents the set of agents, and $W = \{w_1,...,w_n\}$, with $w_i \in \R^{n+1}$ denotes the set of preferences. Here, $w_i \in W$ assigns to each agent a weight, which signifies the preferences of agent $i$. Specifically, $w_{ij}$ reflects the level of satisfaction of agent $i$ with the consumption of her voting weight by $j$ (note that $w_i$ may contain negative entries). The strategy of each agent $i$ is given by its delegation weighting $x_i \in X$. The goal of each agent is to maximize its utility function $U_i: X \to \R$, given by
\begin{equation}\label{eq:utility}
    U_i(x,x_{-i}):= \sum^{n+1}_{j=1} w_{ij} V_j^{\eps}( (x|x_{-i}), \bm{\delta}_i) = w_i \cdot V^{\eps}( (x|x_{-i}), \bm{\delta}_i).
\end{equation}
Here $ \bm{\delta}_i \in \R^{n+1}$ is such that $ \bm{\delta}_{ii} = 1$ and $\bm{\delta}_{ik} = 0$ if $i \not= k$, and $(x|x_{-i})$ denotes the matrix $(x_1,...,x,...,x_n)$ with column $x$ in the $i$-th place, where $x_k$ is the delegation weighting of agent $k$ with $k \not= i$. We take $0<\eps<1$. As illustrated in Section \ref{sec:particles}, $V_j^{\eps}( P, \bm{\delta}_i) \in \R$ quantifies the amount of voting weight consumed by agent $j$ that comes from agent $i$. Hence, $U_i$ can be interpreted as the sum of the voting weight delegated by agent $i$ to other agents, weighted in accordance with her preferences.

To illustrate, consider the example where $N = \{1,2,3\}$, $w_1 = (-2,1,-3,-1)$, $w_2 = (1,-1,0,-2)$, and $w_3 = (0,0,1,0)$, with $\eps = 0$. In this scenario, agent 1 prefers agent 2 over herself, agent 2 favors agents 1 and 3 over herself, and agent 3 prefers herself over the other options. Additionally, agent 1 would rather have her vote reach the null or indecision state than delegate it to agent 3. Given this setup, a reasonable initial strategy for the agents is to delegate all their voting weight to the agent they most prefer. Consequently, agent 1 delegates to agent 2, agent 2 delegates to agent 1, and agent 3 retains all her voting weight. Let $P$ denote the delegation matrix corresponding to this configuration; the utility functions can then be computed as follows:

\begin{equation}
\begin{aligned}
U_1(P) & = w_1 \cdot V(P,\bm{\delta}_1) = & (-2,1,-3,-1) \cdot (0,0,0,1) = & -1,\\
U_2(P) & = w_2 \cdot V(P,\bm{\delta}_2) = & (1,-1,0,-2) \cdot (0,0,0,1) = & -2,\\
U_3(P) & = w_3 \cdot V(P,\bm{\delta}_3) = & (0,0,1,0) \cdot (0,0,1,0) = & 1.
\end{aligned}
\end{equation}
It is evident that this configuration does not constitute a Nash equilibrium. Specifically, agent 2 has a better strategy, which involves delegating her vote to agent 3. Consequently, under this revised strategy by agent 2, agent 1's voting weight would be transferred to agent 3, leading agent 1 to adopt a strategy where she retains all her voting weight. A more thorough analysis of this game shows that no Nash equilibrium exists for this set of preferences. By examining a similar but simpler example, we can rigorously demonstrate that the game $\mathcal{G}(N,U,W,0)$ generally lacks equilibrium states.

\begin{lemma}\label{sec:app_V_not_stable}
    The game $\mathcal{G}(N,U,W,0)$ does not necessarily have a Nash equilibrium. 
\end{lemma}
\begin{proof} Consider the example: $N = \{1,2\}$, $w_1 = (0,-1,1)$, and $w_2 = (1,0,-1)$. Indeed, it can be easily verified that no classic delegation strategy pair constitutes an equilibrium state.
On the other hand, considering fractional strategies $x_1 = (1-\eta_1,\eta_1)$ and $x_2 = (1-\eta_2,\eta_2)$, we examine the following cases:
\begin{itemize}
    \item If $0<\eta_1<1$ and $0<\eta_2<1$, it cannot be an equilibrium as $\eta_2=0$ dominates for agent 2. 
    \item Similarly, for $0<\eta_1<1$ and $\eta_2 = 0$, $\eta_1 = 1$ dominates for agent 1. \item If $\eta_2 = 1$, then $\eta_1=0$ dominates for agent 1. 
\end{itemize}
Therefore, no equilibrium exists for $0<\eta_1<1$. Now, let us consider the cases $\eta_1 = 0$ and $\eta_1 = 1$.

\begin{itemize}
    \item For $\eta_1 = 0$, agent 2's dominant strategy is $\eta_2 = 0$, yet agent 1's dominant strategy is $\eta_1 = 1$. 
    \item With $\eta_1 = 1$, any $\eta_2>0$ strategy for agent 2 is optimal. However, with $\eta_2>0$, agent 1's dominant strategy is $\eta_1 = 0$, indicating it cannot be an equilibrium state.
\end{itemize}
\end{proof}

\subsection{Equilibrium with a positive penalty factor}
Through conventional arguments, one can establish that the game $\mathcal{G}(N,U,W,\eps)$ always possesses at least one Nash equilibrium with $\eps > 0$. To achieve this, we begin by proving the following lemma.

\begin{lemma}\label{lem:upositive}
Let $N = \{1,...,n\}$ be a set of agents with delegation weightings $x_1,...,x_n \in X$, $P$ their delegation matrix, and $f \in \R^{n+1}_{\geq 0}$ in such a way that $f_j > 0$ for some $j \in N$. Then $u^{\eps}_j > 0$ for all $\eps > 0$. 
\end{lemma}

\begin{proof}

From the definition of $u^{\eps}$ we have: 
\begin{equation}\label{eq:upositive1}
    \begin{pmatrix} 
	I - (1-\eps)\Pt & \mathbf{0} \\
	-\bm{\varepsilon} & 1 \\
	\end{pmatrix} u^{\eps} = f,
\end{equation}
Since $f_i \geq 0$ for all $i$, we know from Lemma \ref{Lem:upositive} that $u^{\eps} \geq 0$. From \eqref{eq:upositive1} we have 
$$u^{\eps}_j - (1-\eps) \sum_{i \not = j} P_{ji} u^{\eps}_{i} = f_j > 0, \quad \text{ if } j<n+1.$$
It follows that $u^{\eps}_j$ is greater than zero for all $j < n+1$. The scenario where $j=n+1$ can be dealt with in a similar manner.
\end{proof}

The subsequent result demonstrates that $\mathcal{G}(N,U,W,\eps)$ possesses at least one Nash equilibrium for all $0<\eps<1$.

\begin{theorem}\label{teo:nash}
Let $N = \{1,...,n\}$ be a set of agents with preferences $W = \{w_1,...,w_n\} \subset \R^{n+1}$. The game $\mathcal{G}(N,U,W,\eps)$ has a Nash equilibrium for all $0<\eps<1$. 
\end{theorem}

\begin{proof}
For the sake of convenience, we shall treat the elements of $W$ as row vectors, that is to say $W \subset \R^{n \times n+1}$. Let $r_i(x_{-i})$ denote the optimal responses of agent $i$ to the strategies employed by all other agents:
$$r_i(x_{-i}):= \argmax_{x_{i}} U_i(x_{i},x_{-i}).$$
Let $x$ be an element of $\bm{X} := X^{n} = X \times... \times X$, and let the set valued function $r: \bm{X} \to 2^{\bm{X}}$ be defined as $r(x) := r_1(x_{-1}) \times... \times r_n(x_{-n})$. It is noteworthy that any stationary point of $r$ constitutes a Nash equilibrium \cite{fudenberg1991game}. Our aim is to demonstrate the existence of a stationary point of $r$. To this end, we will employ Kakutani's fixed point theorem, which ensures the existence of a fixed point given the fulfillment of the following conditions:
\begin{enumerate}
    \item \label{cond:ccne} $\bm{X}$ is compact, convex, and nonempty,
    \item \label{cond:ne}$r(x)$ is nonempty,
    \item \label{cond:uhct}$r(x)$ is upper hemicontinuous,
    \item \label{cond:convexity} $r(x)$ is convex.
\end{enumerate}
Condition \ref{cond:ccne} is easily fulfilled. Conditions \ref{cond:ne} and \ref{cond:uhct} can be verified by exploiting the continuity of the function $U_i$ for all $i$, and for all $0<\eps<1$, combined with Berge's maximum theorem \footnote{The "Maximum Theorem" states that the solution set to a maximization problem varies upper semicontinuously as the constraint set of the problem changes continuously \cite{berge}.}.

In order to establish \ref{cond:convexity}, it suffices to demonstrate that $r_i(x_{-i})$ is convex for all $i$, based on the definition of $r(x)$. Let $x$ and $x' \in r_i(x_{-i})$ be two strategies. Let $\eta := x'-x$ and $x(t) := x + t\eta$. Our objective is to show that $x(t) \in r_i(x_{-i})$ for all $t \in [0,1]$. Without loss of generality and for the sake of simplicity, we assume that $i = 1$. Referring to the definition given in \eqref{eq:VPeps}, if we define $P(t) := (x(t)|x_{-1})$, the utility function $U_1$ can be expressed as:
$$U_1(x(t),x_{-1}) = w_1 D_t u(t),$$
where $D_t := diag \big( (P_{11}^{\eps}(t),...,P_{nn}^{\eps}(t),P_{n+1 n+1}^{\eps}(t)) \big)$, and $u(t) := (I-\Pt^{\eps}(t))^{-1} \bm{\delta}_1$. 
We observe that $u(t)$ satisfies the equation 
$$	\begin{pmatrix} 
	I - (1-\eps)\Pt(t) & \mathbf{0} \\
	-\bm{\varepsilon} & 1 \\
	\end{pmatrix} u(t) = \bm{\delta}_1,$$
with $\Pt(t) = P(t) - diag \big( (P_{11}(t),...,P_{nn}(t) \big)$. Using the definition of $x(t)$, we can rewrite the former expression as
$$	\left[ \begin{pmatrix} 
	I - (1-\eps)\Pt(0) & \mathbf{0} \\
	-\bm{\varepsilon} & 1 \\
	\end{pmatrix} - \begin{pmatrix} 
	(1-\eps)t\etat & \mathbf{0} \\
	\mathbf{0} & 0 \\
	\end{pmatrix} \right] 
	u(t) = \bm{\delta}_1, $$
where $\etat \in \R^n$, with $\etat_i := \eta_i$ for all $i>1$ and $\etat_1 := 0$. Defining
\begin{equation}\label{eq:nash0}
    A_0 := \begin{pmatrix} 
	I - (1-\eps)\Pt(0) & \mathbf{0} \\
	-\bm{\varepsilon} & 1 \\
	\end{pmatrix} \text{ and } \Delta := \begin{pmatrix} 
	(1-\eps)\etat & \mathbf{0} \\
	\mathbf{0} & 0 \\
	\end{pmatrix}, 
\end{equation}
we have $(A_0 - t\Delta)u(t) = \bm{\delta}_1$, and then 
\begin{equation}\label{eq:nash1}
    u(t) = A^{-1}_0 \bm{\delta}_1 + tA^{-1}_0 \Delta u(t).
\end{equation}
Also, we observe that $D_t = D_0 + tD_{\Delta}$, with 
\begin{equation} \label{eq:D_delta}
D_{\Delta} := \begin{pmatrix} 
	(1-\eps)(x'_1 - x_1)& \mathbf{0} \\
	\mathbf{0} & \mathbf{0} \\
	\end{pmatrix}. 
\end{equation}
On the other hand, let us consider a constant vector $c \in \R^{1 \times n+1}$, with $c_j = C \in \R$ for all $j$, for some given constant $C$ to be chosen latter. If we define $w^* = w_1 + c$, then it can be observed that maximizing (or minimizing) the utility function $w_1 D_t u(t)$ is equivalent to maximizing (or minimizing) the utility function $U^*(x(t),x_{-1}) := w^* D_t u(t)$. In fact, we have $w^* D_t u(t) = w_1 D_t u(t) + c D_t u(t)$, and 
$$c D_t u(t) = c \cdot V^{\eps}( (x(t)|x_{-1}), \bm{\delta}_1) = C \sum^{n+1}_{j=1} V_j^{\eps}( (x(t)|x_{-1}), \bm{\delta}_1) = C,$$
where in the last step we use Theorem \ref{teo:conservation}. Then $U^*(x(t),x_{-1}) = U_1(x(t),x_{-1}) + C$, and therefore
\begin{equation}\label{eq:nash4}
    \argmax_{x} U_1(x,x_{-1}) = \argmax_{x} U^*(x,x_{-1}).
\end{equation}

Now, we choose $C = -w_{11}$, so that $w^*_1 = 0$. As a consequence, from the definition of $w^*$ and \eqref{eq:D_delta}, we have $w^* D_{\Delta} = \bm{0}$, thus
$$U^*(x(t),x_{-1}) = w^* D_t u(t) = w^* (D_0 + t D_{\Delta}) u(t) = w^* D_0 u(t),$$
and using \eqref{eq:nash1} we obtain
\begin{equation}\label{eq:nash2}
    U^*(x(t),x_{-1}) = w^* D_0 A^{-1}_0 \bm{\delta}_1 + t w^* D_0 A^{-1}_0 \Delta u(t).
\end{equation}
From the definition of $\Delta$ \eqref{eq:nash0} we have
$$\Delta u(t) = (1-\eps)u_1(t) \begin{pmatrix} 
	\etat \\
	 0 \\
	\end{pmatrix},$$
and we can rewrite \eqref{eq:nash2} as follow
\begin{equation}\label{eq:nash3}
    U^*(x(t),x_{-1}) = w^* D_0 A^{-1}_0 \bm{\delta}_1 + t (1-\eps) u_1(t) w^* D_0 A^{-1}_0 \begin{pmatrix} 
	\etat \\
	 0 \\
	\end{pmatrix}.
\end{equation}
Since $x$ and $x'$ are maximizers of $U_1$, from \eqref{eq:nash4} we know that they are also maximizers of $U^*$, then
\begin{align*}
  U^*(x(0),x_{-1}) &= U^*(x(1),x_{-1}), \\
  w^* D_0 A^{-1}_0 \bm{\delta}_1 &= w^* D_0 A^{-1}_0 \bm{\delta}_1 + (1-\eps) u_1(1) w^* D_0 A^{-1}_0 \begin{pmatrix} 
	\etat \\
	 0 \\
	\end{pmatrix},\\
	0 &= (1-\eps) u_1(1) w^* D_0 A^{-1}_0 \begin{pmatrix} 
	\etat \\
	 0 \\
	\end{pmatrix}.
\end{align*}	
From Lemma \ref{lem:upositive}, we know that $u_1(1)>0$. Consequently, we have 
$$
w^* D_0 A^{-1}_0 \begin{pmatrix} 
	\etat \\
	 0 \\
	\end{pmatrix} = 0,
$$
which, by using \eqref{eq:nash3}, leads to $U^*(x(t),x_{-1}) = U^*(x(0),x_{-1})$ for all $t \in [0,1]$. Therefore, $x(t)$ maximizes $U^*$, and from \eqref{eq:nash4}, we conclude that $x(t)$ maximizes $U_1$ for all $t \in [0,1]$. Hence, $r_1(x_{-1})$ is a convex set. Similarly, $r_i(x_{-i})$ is a convex set for all $1 \leq i \leq n$. This completes the proof of the theorem.
\end{proof}

\begin{remark}\label{rem:reulatiry}
    In the proof of Theorem \ref{teo:nash}, two technical factors can be discerned that facilitate the existence of equilibrium states in the game $\mathcal{G}(N,U,W,\eps)$. The main factor is the continuity of $V^\eps$. Specifically, the function $V^\eps: \calP^n \to \R^{n+1}$, unlike its limiting case $V$, is continuous. This continuity is inherited by the agents' utility functions, allowing the use of classical game theory tools to prove the existence of equilibrium states. The second factor is the inherent structure of the utility function defined in \eqref{eq:utility}, which is leveraged to show that the set of optimal responses for an agent is convex.
\end{remark}

\subsection{Relation with similar problems in the literature} \label{sec:EXISTENCE}

We note that if we restrict agents' strategies to convex, compact, and non-empty spaces, Theorem \ref{teo:nash} can be trivially reproduced. For example, we can consider a social network in which $\calN_i \subseteq N$ denotes the neighbors of the agent $i$, and restrict the strategies of each agent to its neighborhood, that is $x_i \in X_i := \{ x \in \R^{n}_{\geq 0} : x_k=0 \,\,\forall k \not \in \calN_i, \text{ and } \sum_k x_k = 1\}$. For this case the proof of the Theorem \ref{teo:nash} can be easily replicated using $\bm{X} := X_1 \times... \times X_n$. This setting can be seen as a generalization of problem \textbf{EXISTENCE}, from \cite{game}[Section 3.2].

In \cite{game}, the authors address the problem of identifying delegation configurations that form a Nash equilibrium within the classical LD paradigm, which permits only simple delegations. The utility functions of each agent are based on a ranking system, where the agent lists, in order of preference, those agents they prefer to vote on their behalf. This ranking system is referred to by the authors as a preference profile, representing a discrete analogue of the preference profiles $W$ defined in this section.

Within the framework of our work, the \textbf{EXISTENCE} problem can be precisely described as follows: Let $N = \{1, \ldots, n\}$ be a set of agents with preferences $W = \{w_1, \ldots, w_n\} \subset \R^{n+1}$, where $w_i \neq w_j$ for all $i, j \in N$, and let $f = \bm{1}$. Define the game $\mathcal{G}^c(N, U, W)$ using the utility function defined in \eqref{eq:utility} with $V^c$ in place of $V^\eps$. Note that in this setting, agents' strategies are limited to classical delegations. The \textbf{EXISTENCE} problem is then defined as determining the existence of Nash equilibria in the game $\mathcal{G}^c(N, U, W)$.

The authors present a simple example demonstrating that a Nash equilibrium does not always exist in this game. In contrast, by employing the measure $V^\eps$, one can ensure the existence of equilibrium states.

On the other hand, the aforementioned problem is also discussed in \cite{pirates}. In that work, the equilibrium states with mixed strategies are examined, and the authors establish the bicriteria approximation guarantees of the game with respect to both rationality and social welfare. However, mixed strategies do not serve as a good generalization of the classical model when deployed to assess the voting weight of agents (see Appendix \ref{subsec:apMS}).

\subsection{A stable Liquid Democracy model}

As we mentioned before, we can demonstrate that employing the measure $V^\eps$ in a liquid democracy system guarantees the existence of equilibria, while ensuring that desirable properties are maintained within a certain error margin that depends on the penalty parameter $\eps > 0$. This observation is formalized in the following result. 

\begin{theorem}\label{teo:main}
Consider a system consisting of agents $N = \{1,2,...,n\}$, where each agent $i \in N$ delegates their vote by using a delegation weighting $x_i \in X$. Let $P \in \mathcal{P}^n$ represent the delegation matrix formed from these weightings. Let $f \in \mathbb{R}_{\geq 0}^{n+1}$, where $f_i$ represents the inherent voting weight of agent $i$. Let $W = \{w_1,w_2,\dots,w_n\} \subset \mathbb{R}^{n+1}$ denote the preferences of these agents over others, as specified in Theorem \ref{teo:nash}. Given $0<\eps<1$, using the voting weight measure $V^{\eps}$ defined in \eqref{eq:VPfeps} the system has the following properties:

    \begin{itemize}
        \item[(I)] $\delta$-Generalization property.
        Given $\delta > 0$ and $f = \bm{1}$, there exists a positive $\eps_\delta$ such that for any $\eps < \eps_\delta$, $\max_{1\leq i \leq n}|V^\eps_i(P) - V_i^c(P)| < \delta$ for all $P \in \calA^n$, where $V^c$ denotes the classic LD measure as defined in \eqref{eq:classicLD}.
        \item[(II)] $\delta$-Delegation property, as introduced in Theorem \ref{teo:delta_delegation}, which asserts the existence of delegation property $(P2)$ subject to a degree of error that is contingent upon the penalty $\eps$.
        \item[(III)] $\delta$-Self-selection. For all $i \in N$ such that $P_{ii} = 0$, it holds that $V_i^\eps(P,f) = 0$. Moreover, given $\delta>0$, there exists $\eps_\delta>0$ such that for all $\eps<\eps_\delta$ and for all $P\in \calP^n$, we have $(1-\delta) f_iP_{ii} \leq V_i^{\eps}(P,f)$.
        \item[(IV)] Conservation of voting weight. Namely, $\sum^{n+1}_{i=1} V^\eps_i(P,f) = \sum^{n+1}_{i=1} f_i$ for all $P \in \calP^n$. 
        \item[(V)] Existence of equilibrium states. The game $\mathcal{G}(N,U,W,\eps)$ has Nash equilibria. 
    \end{itemize}
\end{theorem}

\begin{proof}
The property (I) may be straightforwardly deduced from Theorem \ref{teo:propP1delta}. On the other hand, the properties (II), (III), and (IV) are readily derivable from theorems \ref{teo:delta_delegation}, \ref{teo:propP3}, and \ref{teo:conservation}, respectively. Finally, property (V) can be inferred from Theorem \ref{teo:nash}. 
\end{proof}

\section{Conclusions}

This paper presents a novel method for computing the voting weight of agents in the liquid democracy paradigm. The proposed method extends the classical model by letting agents fractionate and delegate their voting weight to multiple agents, allowing for the imposition of a penalty factor on the length of the delegation chains. This generalization is achieved while preserving the most reasonable properties of the classical model. 

We explore the existence of equilibrium states in the proposed liquid democracy system. Specifically, we define a game in which each agent's delegation weighting constitutes their strategy, and the goal of each agent is to maximize a utility function that quantifies their satisfaction regarding their voting weight. By leveraging the voting weight measure outlined in the first part of this paper and imposing a non-zero penalty factor on the length of delegation chains, we prove that the game admits Nash equilibria. Nonzero penalties facilitate the existence of equilibrium states through regularization of the utility function (see Remark \ref{rem:reulatiry}).

Finally, we present Theorem \ref{teo:main}, which summarizes the obtained results of this work. Our results demonstrate that by introducing a penalty on the length of delegation chains that is suitably small, it is possible to construct a liquid democracy system that maintains the main features of the classical model within a prescribed error margin. Furthermore, such a system allows for the existence of Nash equilibria. To the best of our knowledge, this study provides the first formal proof that establishes the feasibility of a liquid democracy system that guarantees the existence of equilibrium states.

\appendix
\section{}

\begin{proof}[Proof. (of Theorem \ref{teo:conservation} )]\label{proof:teo_cons}
We consider $V^{\eps}(P,f) \in \R^{n+1}$, defined as in \eqref{eq:VPfeps}, for some $\eps>0$. Firstly, our aim is to demonstrate that the equality $\sum^{n+1}_{i=1} V^{\eps}_i(P,f) = \sum^{n+1}_{i=1} f_i$ holds for all $\eps > 0$. To this end, using the definitions provided in equations \eqref{eq:VPfeps} and \eqref{eq:Weps}, we can write 
\begin{equation}\label{eq:consl1}
(I - \Pt^{\eps})u^{\eps} = f.
\end{equation}
Let $D \in \R^{n+1\times n+1}$ be the diagonal matrix $D := diag \big( (P_{11}^{\eps},...,P_{nn}^{\eps},P_{n+1 n+1}^{\eps}) \big) $. Since $\Pt^{\eps} = P^{\eps} - D$, from \eqref{eq:consl1} we obtain 
$$(I - \Pt^{\eps})u^{\eps} = u^{\eps} - P^{\eps}u^{\eps} + Du^{\eps} = f.$$
By summing all the equations of the preceding linear system, we obtain
$$\sum^{n+1}_{i=1}u_i^{\eps} - \sum^{n+1}_{i=1} \sum^{n+1}_{j=1} P_{ij}^{\eps}u_j^{\eps} + \sum^{n+1}_{i=1} V^{\eps}_i(P,f) = \sum^{n+1}_{i=1} f_i.$$
Since $P^{\eps}$ is a stochastic matrix, $\sum^{n+1}_{i=1} P_{ij}^{\eps} = 1$ for all $j$, and then
$$\sum^{n+1}_{i=1} \sum^{n+1}_{j=1} P_{ij}^{\eps}u_j^{\eps} = \sum^{n+1}_{j=1} u_j^{\eps}\sum^{n+1}_{i=1} P_{ij}^{\eps} = 
\sum^{n+1}_{j=1} u_j^{\eps}.$$
By combining this with the previous equality, we obtain:
$$ \sum^{n+1}_{i=1} V^{\eps}_i(P,f) = \sum^{n+1}_{i=1} f_i. $$

Given that $W := \{ w \in \R^{n+1} : \sum^{n+1}_{i=1}w_i = \sum^{n+1}_{i=1} f_i \}$ a closed set and $V^{\eps}(P,f) \in W$ for all $\eps>0$, it follows that $V(P,f) = \lim_{\eps \to 0} V^{\eps}(P,f)$, and therefore, $V(P,f) \in W$. This concludes the proof.

\end{proof}

\begin{proof}[Proof. (of Theorem \ref{teo:propP1} )]\label{proof:teo_propP1}
Let $P \in \mathcal{B}^n$ and $N = \{1,\dots,n\}$ be the set of agents associated with this delegation matrix. From the definitions of $V$ and $V^c$, we observe that $V_i(P) = V^c_i(P) = 0$ for every agent $i$ that belongs to a delegation cycle (according to Definition \ref{def:cycle}). Therefore, without loss of generality, we can assume that the delegation matrix $P$ does not contain cycles (removing the agents of $N_3$ may leave some columns of $P$ with sum less than one, which does not affect the argument below). This allows us to omit the elementwise product in the definition of $V^c$ \eqref{eq:classicLD}. Starting from the alternative definition of $V$ given in \eqref{eq:alter_def}, it suffices to prove that
    \begin{equation}\label{eq:equiv}
            \lim_{k\to\infty} P^k \bm{1} = \lim_{k \to \infty} \big( \sum^{k}_{\ell=0} \Pt^\ell \bm{1} \big) \odot diag(P), \quad \forall P \in \calB^n. 
    \end{equation}
To this end, without loss of generality, we can assume that the candidate agents are listed last. In this way, the matrix $P$ can be written as
\begin{equation*} 
    P = \begin{pmatrix} 
	P_1 & \bm{0} \\
	P_2 & I \\
\end{pmatrix},
\end{equation*}
with $P_1$ being a sub-stochastic matrix. As in \eqref{eq:block_power}, the powers of $P$ can be computed as
\begin{equation*} 
    P^k = \begin{pmatrix} 
	P_1^k & \bm{0} \\
	B_k & I \\
\end{pmatrix},
\end{equation*}
with $k \in \N$, and $B_k = \sum^{k-1}_{\ell=0} I^\ell P_2 P_1^{k - \ell -1} = \sum^{k-1}_{\ell=0} P_2 P_1^\ell$. Applying the same idea, we can see that
\begin{equation*} 
    \Pt = \begin{pmatrix} 
	P_1 & \bm{0} \\
	P_2 & \bm{0} \\
\end{pmatrix},
\end{equation*}
and the powers of this matrix can be computed as
\begin{equation*} 
    \Pt^\ell = \begin{pmatrix} 
	P_1^\ell & \bm{0} \\
	P_2P_1^{\ell-1} & \bm{0} \\
\end{pmatrix},
\end{equation*}
with $\ell \in \mathbb{N}$. 

Now, we have
\begin{equation}\label{eq:equiv0}
P^k \bm{1} =
\begin{pmatrix} 
	P_1^k & \bm{0} \\
	\sum^{k-1}_{\ell=0} P_2 P_1^\ell & I \\
\end{pmatrix}\bm{1} =
\begin{pmatrix} 
	P_1^k \bm{1} \\
	(\sum^{k-1}_{\ell=0} P_2 P_1^\ell)\bm{1} + \bm{1} \\
\end{pmatrix}.
\end{equation}
Also we have
\begin{equation}\label{eq:equiv1}
\sum^{k}_{\ell=0} \Pt^\ell \bm{1} = \bm{1} + \sum^{k}_{\ell=1} \Pt^\ell \bm{1} = \bm{1} + \begin{pmatrix} 
	\sum^{k}_{\ell=1} P_1^\ell & \bm{0} \\
	\sum^{k}_{\ell=1} P_2 P_1^{\ell-1} & \bm{0} \\
\end{pmatrix}\bm{1} = 
\begin{pmatrix} 
	(\sum^{k}_{\ell=0} P_1^\ell) \bm{1} \\
	(\sum^{k-1}_{\ell=0} P_2 P_1^\ell)\bm{1} + \bm{1} \\
\end{pmatrix}.
\end{equation}
Note that $\lim_{k\to \infty} P_1^k \bm{1} = 0$ since $P_1$ is a sub-stochastic matrix. This, together with \eqref{eq:equiv0}, yields
\begin{equation}\label{eq:equiv2}
\lim_{k \to \infty} P^k \bm{1} =
\begin{pmatrix} 
	\lim_{k\to\infty}P_1^k \bm{1} \\
	\lim_{k\to\infty}(\sum^{k-1}_{\ell=0} P_2 P_1^\ell)\bm{1} + \bm{1} \\
\end{pmatrix} = 
\begin{pmatrix} 
	\bm{0} \\
	\lim_{k\to\infty}(\sum^{k-1}_{\ell=0} P_2 P_1^\ell)\bm{1} + \bm{1} \\
\end{pmatrix}.
\end{equation}
We now observe that by listing the candidates last, we have $diag(P) = \begin{pmatrix}
\bm{0} \\
\bm{1} \\
\end{pmatrix}$, combining this with \eqref{eq:equiv1} gives us
\begin{equation}\label{eq:equiv3}
 \lim_{k \to \infty} (\sum^{k}_{\ell=0} \Pt^\ell \bm{1}) \odot diag(P) = 
 \begin{pmatrix} 
	\bm{0} \\
	\lim_{k\to\infty}(\sum^{k-1}_{\ell=0} P_2 P_1^\ell)\bm{1} + \bm{1} \\
\end{pmatrix}.
\end{equation}
Finally, from \eqref{eq:equiv2} and \eqref{eq:equiv3}, we can deduce \eqref{eq:equiv}.
\end{proof}

\begin{proof}[Proof.( of Theorem \ref{teo:propP1delta} )]\label{proof_teo_P1delta}
Since $\calA^n$ is a finite set, the statement follows directly from Theorems \ref{teo:limit} and \ref{teo:propP1}; we nevertheless give a direct argument. We posit the assumption that not all the agents within the system are part of a cycle, otherwise, there would be nothing to demonstrate. By applying Theorem \ref{teo:propP1}, we obtain $|V^\eps_i(P) - V_i^c(P)| = |V^\eps_i(P) - V_i(P)|$ for all $P \in \calA^n$. Therefore, it suffices to establish that the convergence addressed in Theorem \ref{teo:limit} uniformly holds in $\calA^n$.
The proof of this result reveals that the convergence of $V^\eps_i$, where $1\leq i \leq n$, relies on the powers of the matrix $P_1 \in \R^{r \times r}$, $1 \leq r \leq n$, defined in \eqref{eq:numeric}. It is noteworthy that this matrix can be represented as:
\begin{equation*} 
    P_1 = \begin{pmatrix} 
       \bm{0} & B \\
\end{pmatrix},
\end{equation*}
with $B_{ij} \in \{0,1\}$, and $\bm{0} \in \R^{r \times s}$ with $s \geq 1$, since we are considering $P \in \calA^n$. Furthermore, in view of the fact that $P_1$ accounts for interactions among agents who not belong to any delegation cycle, it can be deduced based on the construction of $P_1$ that this matrix is nilpotent. Given that $P_1 \in \R^{r \times r}$ with $r \leq n$, it follows that $\sum_{k\geq0} (1 - \eps)^k P^k_1 \bm{1} = \sum_{0 \leq k \leq n+1} (1 - \eps)^k P^k_1 \bm{1}$. Thus, the uniform convergence of $V^\eps_i(P)$ in $\mathcal{A}^n$ can be readily inferred. 

\end{proof}

\subsection{Spectral LD generalization}\label{app:SG}
We shall denote the generalized Spectral LD measure as the natural extension of the function $V^c$ defined in \eqref{eq:classicLD} to the set $\calP^n$. Specifically, we consider $V^c: \calP^n \to \R_{\geq 0}$,

\begin{equation}\label{eq:classicLDG}
    V^c(P) := \lim_{k\to \infty} ( P^k \mathbf{1} ) \odot diag(P).
\end{equation}

As discussed in Section \ref{sub:problem_mutidel}, the function under consideration does not preserve self-selection (property (P3)). Additionally, it is apparent that this generalization does not fulfill the delegation property (P2). This can be demonstrated by examining the example presented in Figure \ref{fig:example2}.
Considering this, we have
$$
P_l = \begin{pmatrix} 
	  \frac{1}{2} & \frac{1}{2} & 0 \\
	\frac{1}{2} & 0 & 1 \\
           0 & \frac{1}{2} & 0 \\
\end{pmatrix},\text{ and }
P_r = \begin{pmatrix} 
	  \frac{1}{2} & \frac{1}{2} & 1 \\
	  \frac{1}{2} & 0 & 0 \\
                  0 & \frac{1}{2} & 0 \\
\end{pmatrix},
$$
where the delegation matrix $P_l$ corresponds to the left setting, whereas $P_r$ corresponds to the right setting.
\begin{figure}[ht]
\begin{center}
	\begin{tikzpicture}[->, >=stealth', auto, semithick, node distance=1cm]
	\tikzstyle{every state}=[fill=white,draw=black,thick,text=black,scale=0.7]
	\node[state] (B) at (-1,0.5) {$1$};
	\node[state] (C) at (0.5,2) {$2$};
	\node[state] (C2) at (2,0.5) {$3$};
	\node[state] (E) at (3.5,0.5) {$1$};
	\node[state] (F) at (5,2) {$2$};
	\node[state] (F2) at (6.5,0.5) {$3$};
	\path

        (B) edge[loop left,above]	node{$\frac{1}{2}$}	(B)
	(B) edge[bend left,above]	node{$\frac{1}{2}$}	(C)
	(C2) edge[bend right,above]	node{$1$}	(C)
        (C) edge[bend right, below]	node{$\frac{1}{2}$}	(C2)
        (C) edge[bend left,below]	node{$\frac{1}{2}$}	(B)
        (E) edge[loop left,above]	node{$\frac{1}{2}$}	(E)
	(E) edge[bend left,above]	node{$\frac{1}{2}$}	(F)
	(F2) edge[below]	node{$1$}	(E)
        (F) edge[above]	node{$\frac{1}{2}$}	(F2)
        (F) edge[below]	node{$\frac{1}{2}$}	(E);
	%\node[above=0.5cm] (A){Patch G};
	%\draw[red] ($(D)+(-1.5,0)$) ellipse (2cm and 3.5cm)node[yshift=3cm]{Patch H};
	\end{tikzpicture}
\end{center}
\caption{We consider two settings with $N = \{1,2,3\}$, where the delegation property establishes that the voting weight distribution should be the same for both settings. However, it can be verified that this distribution differs when evaluated in each of the aforementioned settings using the measure $V^c$.}
\label{fig:example2}
\end{figure}
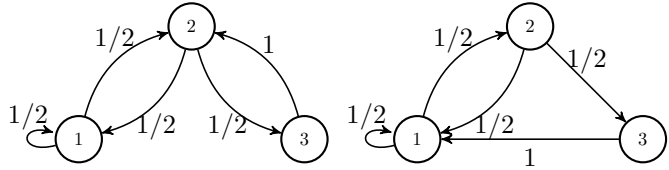
By applying \eqref{eq:classicLDG}, it is readily apparent that $V^c(P_l) = (\frac{3}{5},0,0)$ and $V^c(P_r) = (\frac{6}{7},0,0)$. Note that the computation in this case can be implemented by simply calculating the equilibrium states of the Markov chain defined by the delegation matrices and then multiplying the resulting vector by the diagonal of the matrix.
 Therefore, it can be deduced that $V^c$ fails to satisfy the delegation property. 

Notably, since the voting measures proposed in \cite{castillo,viscousDemocracy,yoshida} are closely related to the generalization described in \eqref{eq:classicLD}, it can be shown, using ideas similar to those outlined above, that the voting measures in these works do not satisfy properties (P2) and/or (P3).

%It is worth noting, however, that $V(P_l) = V(P_r) = (3,0,0,0)$.

\subsection{Mixed strategies as multi-agent delegation}\label{subsec:apMS} As mentioned earlier, in \cite{pirates}, the authors examine a game analogous to the one defined in Section \ref{sec:nash} within the framework of the classic model of Liquid Democracy. There, they establish that the game has bicriteria approximation guarantees concerning both rationality and social welfare when mixed strategies are taken into consideration. However, mixed strategies are not a suitable generalization of the classical model when employed to assess the voting weight of agents.

To see this, we consider a system consisting of agents $N = \{1,...,n\}$, where $x_1,...,x_n \in X$ represent their respective delegation weightings, and $P \in \calP^n$ is the delegation matrix derived from these weightings. Here $\calA^n = \{ A \in \calP^n : A_{ij} \in \{0,1\} \text{ for all } 1 \leq i,j \leq n \}$. We define the function $p: \calA^n \times \calP^n \to [0,1]$.
$$
p(A,P) := \prod^n_{i=1} A^t_i \cdot P_i, 
$$
where, $A_i$ and $P_i$ correspond to the $i$-th column of $A$ and $P$, respectively. It is easy to see that $p(\cdot, P)$ defines a discrete probability density function on the set $\calA^n$. We can therefore define the mixed strategy-based voting weight measure $V^{MS}: \calP^n \to \R^n$ as follows:
\begin{equation}\label{eq:VMS}
    V^{MS}(P) := \sum_{A \in \calA^n} p(A,P)V^C(A).
\end{equation}
Here, $V^C$ represents the voting weight measure of the classic model defined in \eqref{eq:classicLD}.

We can easily observe that the voting weight measure $V^{MS}$ does not satisfy the delegation property (P2). To demonstrate this, we can consider the example illustrated in Figure \ref{fig:example1}, where the delegation matrices are given by:
$$
P_l = \begin{pmatrix} 
	  1 & \frac{1}{2} & 0 & 0\\
	0 & 0 & 1 & 0\\
        0 & \frac{1}{4} & 0 & 0\\
        0 & \frac{1}{4} & 0 & 1\\
\end{pmatrix},\text{ and }
P_r = \begin{pmatrix} 
	  1 & \frac{1}{2} & \frac{2}{3} & 0\\
	0 & 0 & 0 & 0\\
        0 & \frac{1}{4} & 0 & 0\\
        0 & \frac{1}{4} & \frac{1}{3} & 1\\
\end{pmatrix},
$$
for the left and right settings, respectively. Formula \eqref{eq:VMS} yields $V^{MS}(P_l) = (2,0,0,\frac{3}{2})$ and $V^{MS}(P_r) = (\frac{7}{3},0,0,\frac{5}{3})$, leading to the conclusion that $V^{MS}$ violates the delegation property.

\bibliographystyle{abbrv}
\bibliography{delegative.bib}

\end{document}